\documentclass[12pt]{article}

\usepackage{amsmath}

\usepackage{multirow} 
\usepackage{booktabs}
\usepackage{bbold}
\usepackage{bm}
\RequirePackage{amsthm,amsmath,amsfonts,amssymb}
\RequirePackage[numbers]{natbib}
\RequirePackage[colorlinks,citecolor=blue,urlcolor=blue]{hyperref}
\RequirePackage{graphicx}

\newtheorem{theorem}{Theorem}[section]

\newtheorem{proposition}[theorem]{Proposition}
\newtheorem{lemma}[theorem]{Lemma}

\theoremstyle{remark}

\usepackage[title]{appendix}
\usepackage{algorithmic}
\usepackage[ruled,vlined]{algorithm2e}
\begin{document}

\title{An Online Projection Estimator for Nonparametric Regression in Reproducing Kernel Hilbert Spaces}

\author{Tianyu Zhang \and Noah Simon}

\maketitle

\begin{abstract}
The goal of nonparametric regression is to recover an underlying regression function from noisy observations, under the assumption that the regression function belongs to a pre-specified infinite dimensional function space. In the online setting, when the observations come in a stream, it is generally computationally infeasible to refit the whole model repeatedly. There are as of yet no methods that are both computationally efficient and statistically rate-optimal. In this paper, we propose an estimator for online nonparametric regression. Notably, our estimator is an empirical risk minimizer (ERM) in a deterministic linear space, which is quite different from existing methods using random features and functional stochastic gradient. Our theoretical analysis shows that this estimator obtains rate-optimal generalization error when the regression function is known to live in a reproducing kernel Hilbert space. We also show, theoretically and empirically, that the computational expense of our estimator is much lower than other rate-optimal estimators proposed for this online setting.
\end{abstract}

\section{Introduction}
 It is often of interest to estimate an underlying regression function, linking features to an outcome, from noisy observations. In the case that the structure of this function is not known (e.g. when we do not want to assume a simple linear form), some form of nonparametric regression is employed. More formally, suppose we observe $(X_i,Y_i)\stackrel{i.i.d.}{\sim}\rho(X,Y)$, $i=1,2,...,n$ generated from the following statistical model:
\begin{equation}\label{reg}
Y_i = f_{\rho}(X_i)+\epsilon_i
\end{equation}
where, for each $i$, $X_i\stackrel{i.i.d.}{\sim}\rho_X$ (which take value in $\mathbb{R}^d$) are our features, $Y_i\in\mathbb{R}$ is our outcome, $\epsilon_i$ are iid mean $0$ noise variables. One can think of $f_{\rho}$ as implicitly defined by the joint distribution $\rho(X,Y)$. It is often of interest to estimate $f_{\rho}$, the regression function (e.g. in predictive modeling, or inferential applications). Under mild conditions, the regression function $f_{\rho}$ can also be characterized as the minimizer of
\begin{equation}\label{regression}
    \min_{f\in\mathcal{F}} \mathbb{E}{(Y-f(X))^2}
\end{equation}
when $\mathcal{F} = L^2_{\rho_X}$, which is the best measurable function for predicting $Y$ given $X$ under least squares loss.

\subsection{Nonparametric Regression in RKHS}
In nonparametric regression we often assume that $f_{\rho}$ belongs to a specified infinite dimensional function space $\mathcal{F}$. This is known as the \emph{Hypothesis Space}. Some commonly used $\mathcal{F}$ in statistics and computer science communities are the Holder ball, Sobolev spaces \citep{wahba1990spline}, general reproducing kernel Hilbert spaces (RKHS) \citep{christmann2008support}, and Besov spaces \citep{hardle2012wavelets}. In this paper, we focus on estimation when $\mathcal{F}$ is a RKHS. Briefly, a RKHS over $\mathcal{X}$ is a Hilbert space $(\mathcal{F},\langle\cdot,\cdot\rangle_{\mathcal{F}})$ with the reproducing property: for any $f\in\mathcal{F}$, $x\in\mathcal{X}$, 
\begin{equation}
\label{eq:reproducing}
    f(x)=\left\langle f, K_{x}\right\rangle_{\mathcal{F}}
\end{equation}
where $K_x$ is the so-called kernel function associated with $\mathcal{F}$ evaluated at $x$. This is discussed in more detail in Section~\ref{section:preliminaries}.

In the classical non-streaming setting of nonparametric regression, estimation in a RKHS $\mathcal{F}$ is a well-studied problem. In this case, the kernel ridge regression (KRR) estimator is the gold standard, e.g. \citep{wainwright2019high}. It is defined by
\begin{equation}
\label{KRR}
\hat f_{n}^{KRR}:=\underset{f \in \mathcal{F}}{\mathrm{argmin}}\frac{1}{n} \sum_{i=1}^{n}\left(Y_{i}-f\left(X_{i}\right)\right)^{2}+\lambda_{n}^{KR R}\|f\|_{\mathcal{F}}^{2}
\end{equation}
where $\lambda_{n}^{KR R}$ is a hyper-parameter that balances the mean-square error and the complexity of estimate. Thanks to the reproducing property \eqref{eq:reproducing}, $\hat f_n^{KRR}$ can be written as a finite linear combination of the kernel function evaluated at $(X_i)_{i=1}^{n}$ \citep{scholkopf2001generalized}.

In general \eqref{KRR} requires solving an $n\times n$ linear system, and thus, will have a computational expense on the order of $n^3$. In the online setting, this is exacerbated by the need to refit for each new observation resulting in $n^4$ computation required to fit a sequence of $n$ estimators. While this penalized estimator has good statistical properties (rate optimal convergence and strong empirical performance), the high computational expense restricts its application in the online setting. Substantial effort has been spent on reducing the computational expense of KRR using, for example, "scalable kernel machines" based on random Fourier feature (RFF) \citep{liu2020random} or Nyström projection \citep{gittens2016revisiting}. This is further discussed in Section~\ref{section:citation}.

\subsection{Parametric and Nonparametric Online Learning}

Online learning has been thoroughly studied in the parametric setting: there we assume $f_{\rho}$ takes a parametric form indexed by a finite dimensional parameter $\beta\in\mathbb{R}^p$ (e.g. $f_{\rho}(X) = \beta^{\top}X$ for a linear model).

In this parametric online setting, it is useful to frame the regression function as a population minimizer
\begin{equation}\label{OP}
\min_{\beta\in\mathbb{R}^p} \mathbb{E}[(Y-f_{\beta}(X))^2]  
\end{equation}
From here, it is popular to directly apply stochastic gradient descent (SGD) to \eqref{OP}, by using each sample in our "stream" to calculate one unbiased estimate of the gradient. Updating such an estimator with a new observation has constant computational expense, $O(p)$. Additionally, these estimators achieve the optimal parametric convergence rate $O(1/n)$ under mild conditions \citep{kushner2003stochastic,bach2013non,frostig2015competing,babichev2018constant}.
 
However, comparatively less attention has been given to online nonparametric regression. A few rate-optimal functional stochastic gradient descent algorithms have been proposed in the last decade \citep{tarres2014online,dieuleveut2016nonparametric}, where the hypothesis function space $\mathcal{F}$ is assumed to be a RKHS. The RKHS structure makes it possible to take the gradient of the evaluation functional $L_x(f):= f(x)$. Although such estimators have been shown to be statistically rate-optimal, updating them with a new observation $(X_{n+1},Y_{n+1})$ usually involves evaluating $n$ kernel functions at $X_{n+1}$, with computational expense of order $O(n)$. This is in contrast with the constant update cost of $O(p)$ in parametric SGD. Thus, the computational cost of nonparametric SGD will accumulate at order $O(n^2)$, which is not ideal for methods that are nominally designed to deal with large datasets. Although there has been some effort devoted to transfer RFF- or Nystrom- based methods to the online setting (See Section~\ref{section:citation}), the theoretical guarantees are usually not close to optimal with strong restrictions on the noise variables.

\textbf{Our contribution} In this paper, we propose a method for constructing online estimators in a RKHS by considering the Mercer expansion (eigendecomposition) of a kernel function. Existing methods usually takes an iterative form, which can be interpreted as projecting a random function onto a random space with growing dimension \cite[Equation~(15)]{koppel2019parsimonious}. However, our estimator is the first one that can be treated as an empirical risk minimizer (ERM, or M-estimator of negative loss) in a deterministic linear space with growing dimension.

Analysis of both the statistical and the computational properties of the estimator is performed to show that: i) it has asymptotically optimal (up to a logarithm term) generalization error; ii) it has significantly lower computational expense than other proposed rate-optimal nonparametric SGD estimators; iii) it is robust against heavy-tailed noise. Interestingly, it only requires the $(1+\Delta)$ moment of the noise to be finite for any $\Delta>0$ to achieve consistency.


It is worth noting that in the theoretical analysis of our estimator, we do not require the covariate $X$ to be equally-spaced or uniformly distributed as in standard references \citep{introtononpara} (though such assumptions could significantly simplify the proof). We additionally do not require it to be known for rate optimal convergence. We show that our estimator will obtain rate optimal convergence if $\rho_X$ is absolutely continuous with respect to the measure that is used to conduct the eigendecomposition of the kernel function (usually the latter is taken as uniform measure or a Gaussian distribution).

\textit{Notation:} we use $a_n = \Theta(b_n)$ to indicate that the two sequences increase/decrease at the same rate as $n\rightarrow\infty$. Formally,
\begin{equation}
0<\liminf _{n \rightarrow \infty}\left|\frac{a_{n}}{b_{n}}\right| \leq \limsup _{n \rightarrow \infty}\left|\frac{a_{n}}{b_{n}}\right|<\infty
\end{equation}
For $a\in\mathbb{R}$, $\lfloor a \rfloor$ is the largest integer that is smaller than or equal to $a$. The $\|\cdot\|_2$-norm of a function is its $L^2_{\rho_X}$-norm, i.e $\|f\|^2_2 = \int_{\mathcal{X}}f^2(z)d\rho_X(z)$. In this paper, by saying two functions $f,g$ are orthogonal with respect to measure $P$, we mean $\int f(x)g(x)dP(x) = 0$.
\section{Preliminaries on RKHS}
\label{section:preliminaries}
In this section, we are going to provide background information on RKHS and existing methods before introducing our estimation procedure.

First we formally introduce the concept of Mercer kernel and its corresponding RKHS. A symmetric bivariate function $K:\mathcal{X}\times\mathcal{X}\rightarrow \mathbb{R}$ is positive semi-definite (PSD) if: for any $n\geq 1$ and $(x_i)_{i=1}^n\subset\mathcal{X}$, the $n \times n$ kernel matrix $\mathbb{K}$ whose elements are $\mathbb{K}_{ij}:= K(x_i,x_j)$ is always a PSD matrix. A continuous, bounded, PSD kernel function $K$ is called a \emph{Mercer kernel}. We have the following duality between a Mercer kernel and a Hilbert space:
\begin{proposition}
\label{prop:definitionRKHS}
For any Mercer Kernel $K:\mathcal{X}\times\mathcal{X}\rightarrow\mathbb{R}$, let $K_x$ denote the function $K_x(\cdot):=K(x,\cdot)$. There exists an unique Hilbert Space $(\mathcal{H},\langle \cdot,\cdot \rangle_{\mathcal{H}})$ of functions on $\mathcal{X}$ satisfying the following conditions. 
\begin{enumerate}
    \item For all $x\in \mathcal{X}$, $K_x\in \mathcal{H}$.
    \item The linear span of $\{K_x\ |\ x\in \mathcal{X}\}$ is dense (w.r.t $\|\cdot\|_{\mathcal{H}}$) in $\mathcal{H}$
    \item (reproducing property) For all $f\in \mathcal{H},x\in\mathcal{X}$, 
    \begin{equation}
        f(x) = \langle f,K_x\rangle_{\mathcal{H}}
    \end{equation}
\end{enumerate}
\end{proposition} 
We call this Hilbert space the \emph{Reproducing kernel Hilbert space (RKHS)} associated with kernel $K$, or the \emph{native space} of $K$. For a more comprehensive discussion of RKHS, see \citet{cucker2002mathematical, wainwright2019high, fasshauer2015kernel}.

There is an equivalent definition of RKHS that we will engage with in this manuscript. Given any Mercer kernel $K$ and any Borel measure $\nu$, there exists a set of $L^2_{\nu}$-orthonormal basis $(\phi_j)_{j=1}^{\infty}$ of $\bar{\mathcal{H}}$ (closure of $\mathcal{H}$ with respect to $\|\cdot\|_{L^2_{\nu}}$). Additionally, each of the functions has a paired positive real number $\mu_j$, sorted s.t. $\mu_j \geq \mu_{j+1}>0$. We call the functions $\phi_j$'s eigenfunctions and $\mu_j$'s their corresponding eigenvalues. We state the following equivalent definition of the native space of $K$.
\begin{proposition}
\label{prop:ellipsoidRKHS}
Define a Hilbert space
\begin{equation}
\mathcal{H} = \left\{f \in L^2_{\nu}\mid f=\sum_{k=1}^{\infty} \theta_{j} \phi_{j}\ \text {with} \sum_{j=1}^{\infty} \left(\frac{\theta_j}{\sqrt{\mu_j}}\right)^2< \infty\right\}    
\end{equation}
equipped with inner product:
\begin{equation}
\langle f, g\rangle_{\mathcal{H}}=\sum_{j=1}^{\infty} \frac{a_{j} b_{j}}{\mu_{j}}
\end{equation}
for \(f=\sum_{j=1}^{\infty} a_{j} \phi_{j}\) and \(g=\sum_{j=1}^{\infty} b_{j} \phi_{j}\).

Then $(\mathcal{H},\langle \cdot, \cdot\rangle_{\mathcal{H}})$ is the reproducing Hilbert space of kernel $K$.
\end{proposition}

For discussion of this definition and its relation to Proposition~\ref{prop:definitionRKHS}, see \cite{cucker2002mathematical}. For many kernels, the analytical form of the $(\mu_j,\phi_j)$'s are available for some specific choice of measure $\nu$. This can be quite useful for our method: We require the eigen-system of the kernel with respect to some (relatively arbitrary) measure. This measure does \emph{not} need to be the measure $\rho_{X}$, it merely needs to be absolutely continuous with respect to $\rho_{X}$. In this manuscript we will assume such a convenient measure, denoted, $\bar{\rho}_X$ exists (for which the kernel has an accessible eigen-system and $\bar{\rho}_X \ll \rho_X$). We will call it a \emph{working measure}, and use the notation $(\lambda_j,\psi_j)$ instead of the generic $(\mu_j,\phi_j)$ to denote such an eigen-system with respect to $L^2_{\bar{\rho}_X}$. As an example, the kernel $K(x,z) = \min\{x,z\}$ is the reproducing kernel of Sobolev space
\begin{equation}
\label{eq:W10}
    W_1^0([0,1]) = \left\{f:[0,1]\rightarrow\mathbb{R}\ |\ f(0) = 0\text{ and }\int_0^1 \left(f'(x)\right)^2dx < \infty\right\}
\end{equation}
and its eigenfunctions and eigenvalues are (w.r.t. $\bar{\rho}_X$ = $\text{Unif}([0,1])$):
\begin{equation}
    \psi_j(x) = \sqrt{2}\sin \left(\frac{(2 j-1) \pi x}{2}\right)\quad \lambda_j = \frac{4}{(2 j-1)^{2} \pi^{2}}
\end{equation}

It is also possible to write the kernel as a \emph{Mercer expansion} w.r.t $(\psi_j,\lambda_j)$:
\begin{equation}
\label{kernelexpension}
    K(x,z) = \sum_{j=1}^{\infty} \lambda_j\psi_j(x)\psi_j(z)
\end{equation}
The functions $\{\sqrt{\lambda_j}\psi_j(x), j=1,2,...\}$ are also called the feature maps of the kernel $K$. Also note that by definition $\psi_j$'s are orthogonal w.r.t. $\langle\cdot,\cdot\rangle_{\mathcal{H}}$. There is a collection of $20$ commonly used kernels' Mercer expansion in \cite[Appendix~A]{fasshauer2015kernel}.

If a function $f= \sum_{j=1}^{\infty} \theta_j\psi_j$ has a finite $\|\cdot\|_{\mathcal{H}}$ RKHS-norm. Its general Fourier coefficients $(\theta_j)_{j\in\mathbb{N}}$ need to be at least $o(\lambda_j j^{-1/2})$ so that the norm series $\sum_{j=1}^{\infty} (\theta_j/\sqrt{\lambda_j})^2$ converges. This suggests, for sufficiently large $N$, the truncation $f_N = \sum_{j=1}^{N} \theta_j\psi_j$ should be a good approximation to $f$. This basic idea motivates our work --- by analyzing the spectrum of the kernel we can identify what $N$ should be.

\subsection{Existing Online Nonparametric Methods}
\label{section:citation}
In a RKHS, it is possible to take the functional gradient of the evaluation operator $L_x$ for any $x\in\mathcal{X}$. This allows methods using functional SGD to solve the regression problem \eqref{regression}. Usually \emph{functional SGD} estimators after $n$ steps, $\hat f_n^{SGD}$ of $f_{\rho}$ take the form of a weighted sum of $n$ kernel functions $K_{X_i},i=1,2,...,n$, \citep{tarres2014online,dieuleveut2016nonparametric}:
\begin{equation}
\label{SGD}
\hat f_n^{SGD} = \sum_{i=1}^{n}a_i K_{X_i}   
\end{equation}
To update $\hat f_{n}^{SGD}$ with $(X_{n+1},Y_{n+1})$, it is necessary to evaluate all $n$ kernel basis functions $\{K_{X_i},i=1,2,...n\}$ at $X_{n+1}$. Thus, the computational expense of the update is $O(n)$. There has been some work to improve this computational expense: in \cite{si2018nonlinear, lu2016large, koppel2019parsimonious}, the authors choose a subset of features $(K_{X_i})_{i=1}^n$ whose cardinality is smaller than $n$; in \cite{dai2014scalable, lu2016large}, kernel-agnostic random Fourier features are used: Typically $O(\sqrt{n})$ basis functions are required in this setting, cf. \cite{rudi2017generalization}. Although computationally more efficient than vanilla functional SGD \eqref{SGD}, the theoretical aspects of these scalable methods are not fully satisfying: 1) noise variables are required to have extreme light tails to provably guarantee convergence; 2) Verified convergence rates are not minimax-optimal; and 3) The target parameter is generally not even  $f_{\rho}$ but, instead, a penalized population risk-minimizer. 


Compared with the linear space spanned by random features or kernel functions (engaged with in previous work), the space spanned by eigenfunctions has minimal approximation error in the sense of minimizing Kolmogorov N-width \cite[Section~3]{santin2016approximation}. This inspired us to use them as basis functions to construct our estimator. Briefly, this means projecting onto the N-dimension linear space spanned by the eigenfunctions has the minimal residual, among all the N-dimension linear sub-spaces of $L^2_{\bar{\rho}_X}$. More technically: 
\begin{equation}
 \sup _{\|f\|_{\mathcal{H}} = 1}\left\|f-\Pi_{L^2_{\bar{\rho}_X},\ \mathcal{F}_{N}} f\right\|_{L^2_{\bar{\rho}_X}} = \inf _{V_{N} \subset L^2_{\bar{\rho}_X}} \sup _{\|f\|_{\mathcal{H}} = 1}\left\|f-\Pi_{L^2_{\bar{\rho}_X},\ V_{N}} f\right\|_{L^2_{\bar{\rho}_X}}
 =\sqrt{\lambda_{N+1}}
\end{equation}
where $\mathcal{F}_N$ is the linear space spanned by the first $N$ eigenfunctions $(\psi_j)_{j=1}^N$, $\Pi_{A,B}$ is the projection operator onto space $B$ using the inner product of $A$ and $V_N$ is a generic $N$-dimension linear space in $L^2_{\bar{\rho}_X}$. This is important for statistical estimation as there is a bias/variance tradeoff at play in this estimation problem (more basis functions decreases bias but increases variance). By using a basis that can more compactly represent our function, we can find a more favorable tradeoff and asymptotically decrease our estimation error.

Our research aims to propose a method with favorable statistical guarantees (minimax rate-optimality) and a lower computational expense. The basis functions used should be kernel-sensitive and the convergence rate should be sensitive to the decay rate of eigenvalues $\lambda_j$. Also, we give provable theoretical guarantees in a heavy tail noise setting.

\section{A Computationally Efficient Online Estimator}
In this section we will present the proposed online regression estimator. We first discuss the well-known projection estimator in the batch learning setting, then shift to the online setting where we naively refit the model with each observation; and finally give our proposed modification to make this process computationally efficient. In what follows we will use $N$ to denote the number of basis functions used to construct each projection estimator, though it should more formally be written as $N(n)$, as it is a non-decreasing function of $n$.
\subsection{Projection Estimator in Batch Learning}
Suppose we have $n$ samples $(X_i,Y_i)_{i=1}^n$, and let $\mathcal{F}_N = \text{span}(\psi_1,...,\psi_{N})$ be the $N$-dimension linear space spanned by the $N$ eigenfunctions with largest eigenvalues. The function $\hat f_{n,N}$ that minimizes empirical mean square error over $\mathcal{F}_N$ is a very attractive candidate for estimating $f_{\rho}\in\mathcal{H}$, that we aim to leverage for the online setting. 

Formally, define $\pmb{\theta} = (\theta_1,...,\theta_{N})^{\top}$ and $\pmb{\psi}^N(X_i) = (\psi_1(X_i),...,\psi_{N}(X_i))^{\top}$. Consider the least squares problem (in Euclidean space):
\begin{equation}\label{NPLS}
\min_{\pmb{\theta}\in\mathbb{R}^{N}} \sum_{i=1}^n (Y_i - \pmb{\theta}^{\top} \pmb{\psi}^N(X_i))^2
\end{equation}
The solution can be written in matrix form as
\begin{equation}
 \pmb{\hat\theta} := (\hat \theta_1,...,\hat \theta_{N})^{\top} = (\Psi^{\top}_n \Psi_n)^{-1}\Psi^{\top}_n \pmb{Y}_n   
\end{equation}
if $\Psi_n^{\top}\Psi_n$ is invertible. Here $\pmb{Y}_n = (Y_1,...,Y_n)^{\top}$ is the observed response, and $\Psi_n$ is the design matrix whose elements are $\Psi_{ij} = \psi_j(x_i)$. Then the estimator
\begin{equation}\label{PE}
    \hat f_{n,N} = \sum_{j=1}^N \hat\theta_j \psi_j
\end{equation}
is the empirical risk minimizer (ERM) in $\mathcal{F}_N$. Estimators that take this form are called nonparametric \emph{projection estimators} (of $f_{\rho}$, with level $N$) \citep{introtononpara}.

The optimal number of basis functions to use depends on both the sample size $n$ and how fast the eigenvalues $\lambda_j$ in \eqref{kernelexpension} decay. As we will state formally in Theorem \ref{maintheorem}, the optimal choice is $N = \Theta(n^{\frac{d}{2\alpha+d}})$ when $\lambda_j = \Theta(j^{-2\alpha/d})$, with $\alpha>\frac{d}{2}$. Note that the condition $\alpha>\frac{d}{2}$ ensures the considered RKHS can be embedded into the space of continuous functions (as a result of Sobolev inequality, cf. Theorem~12.55 \citep{leoni2017first}). With this choice for $N$, convergence of $\hat f_{n,N}$ achieves the minimax rate over functions with bounded RKHS norm. Similar results for projection estimators have been shown when $(\psi_j)_{j=1}^{\infty}$ is the trigonometric basis, and $x_i$ are deterministic, and evenly spaced \citep{introtononpara} or $\rho_X$ is the uniform distribution \citep{belloni2014pivotal}. Our analysis shows that the optimality of the projection estimator actually holds for general $\psi_j$ and does not require them to be orthonormal with respect to the empirical measure or $\rho_X$.



\subsection{Naive Online Projection Estimator}

The most direct way of extending the projection estimator \eqref{PE} to the online setting is simply refitting the whole model whenever a new pair of data $(X_i,Y_i)$ comes in. In Algorithm~\ref{alg:naive}, we provide this naive updating rule for our reader to better understand the proposed method. Our modified proposal in Section~\ref{efficientsection} greatly improves upon this in terms of computational cost but giving the same estimates $\hat f_{n,N}$.

\begin{algorithm}[!tb]

   \caption{Naive rule for updating $\hat{\pmb{\theta}}$ with a new observation $(X_n,Y_n)$.}
   \label{alg:naive}
\begin{algorithmic}
   \STATE {\bfseries Input:} $(X_i)_{i=1}^n,\pmb{Y}_n,\Phi_{n-1},\Psi_{n-1}, \alpha, N$
\STATE \textbf{function} UpdateCurrent($X_n, N,\Phi_n,\Psi_n$)
\STATE $\pmb{\psi}_n \gets [\psi_1(X_n),\psi_2(X_n),...,\psi_N(X_n)]^{\top}$
\STATE $\Psi_n\gets \begin{bmatrix} \Psi_{n} \\\pmb{\psi}_{n}^{\top}\end{bmatrix}$ $\quad\Phi_{n} \gets \left(\Psi_n^{\top}\Psi_n\right)^{-1}$
\STATE \textbf{return} $(\Phi_{n},\Psi_n)$
   \STATE \textbf{function} AddBasis ($(X_i)_{i=1}^n, N,\Phi_{n},\Psi_{n}$)
\STATE $\pmb{\psi}^{N+1} \gets [\psi_{N+1}(X_1),...,\psi_{N+1}(X_n)]^{\top}$
\STATE $\Psi_n\gets \begin{bmatrix} \Psi_{n} & \pmb{\psi}^{N+1}\end{bmatrix}$ $\quad\Phi_{n} \gets \left(\Psi_n^{\top}\Psi_n\right)^{-1}$
\STATE \textbf{return} $(\Phi_{n},\Psi_n)$
\IF{$n= \text{Floor}( (N+1)^{2\alpha+1})$}
\STATE $(\Phi_{n},\Psi_n) \gets \text{UpdateCurrent}(X_n, N, \Phi_{n-1},\Psi_{n-1})$
\STATE $(\Phi_{n},\Psi_n) \gets\ \text{AddBasis}((X_i)_{i=1}^n,N,\Phi_{n},\Psi_{n})$
\STATE $N\ \gets\ N+1$
\ELSE
\STATE $(\Phi_{n},\Psi_n) \gets \text{UpdateCurrent}(X_n, N, \Phi_{n-1}, \Psi_{n-1})$
\ENDIF
\STATE$\hat{\pmb{\theta}}\gets \Phi_n \Psi_n^{\top} \pmb{Y}_n$
\end{algorithmic}
\end{algorithm}

In this algorithm, $\pmb{Y}_n = (Y_1,...,Y_n)^{\top}$ is the vector of outcomes. $\Psi_{n}$ is the $n\times N$ design matrix at step $n$, and $\Phi_n$ denotes the $N\times N$ matrix $(\Psi_n^{\top}\Psi_n)^{-1}$ (inversion of Gram matrix). 

Whenever new data comes in, the algorithm augments the design matrix by adding one new row to $\Psi_{n-1}$ based on the new observation $X_n$. The new row $[\psi_1(X_n),\psi_2(X_n),...,\psi_N(X_n)]$ can be understood as the embedding of $X_n$ into the feature space spanned by $(\psi_j)_{j=1}^N$.

When $n = \lfloor (N+1)^{\frac{2\alpha+d}{d}}\rfloor$, this algorithm additionally adds a new column to the design matrix $\Psi_n$ (increasing the dimension of the basis function we project upon by $1$). This new column is just the evaluation of $\psi_{N+1}$ at $(X_i)_{i=1}^n$. Recall that $\psi_{N+1}$ is the $(N+1)$-th eigenfunction in Mercer expansion \eqref{kernelexpension}. It is straightforward to show that this criterion of adding new basis functions ensures $N = \Theta(n^{\frac{d}{2\alpha+d}})$. 

The computational expense of each update using Algorithm~\ref{alg:naive} is $\sim n^{\frac{2\alpha+3d}{2\alpha+d}}$. In particular, calculating $\Psi_n^{\top} \Psi_n $ takes $\sim nN^2 \sim n^{\frac{2\alpha+3d}{2\alpha+d}}$ computation. While this algorithm would give a statistically rate-optimal estimator, and is straightforward to implement, it is rather computationally expensive. In particular, the functional SGD algorithm has a comparatively smaller computational cost of $\sim n$ per update.

\subsection{Efficient Online Projection Estimator}
In this section we explicitly give our proposed method (the details of which are given in Algorithm~\ref{alg:efficient}). By using some common block/rank-one updating tools from linear algebra, we are able to substantially improve Algorithm~\ref{alg:naive}. In particular, it is expensive to repeatedly calculate $(\Psi_n^{\top}\Psi_n)^{-1}$ directly. However, matrix $\Psi_n$ has only one more row and (sometimes) one more column than $\Psi_{n-1}$. It is possible to calculate $(\Psi_n^{\top}\Psi_n)^{-1}$ by updating $(\Psi_{n-1}^{\top}\Psi_{n-1})^{-1}$. The latter will already have been calculated when observing $(X_{n-1},Y_{n-1})$. 

When $\Psi_n$ has one more row than $\Psi_{n-1}$:
\begin{equation}
    \Psi_n = \begin{bmatrix} \Psi_{n-1} \\ \pmb{\psi}_{n}^{\top}
    \end{bmatrix}
\end{equation}
where $\pmb{\psi}_n = \left[\psi_{1}\left(X_{n}\right), \psi_{2}\left(X_{n}\right), \ldots, \psi_{N}\left(X_{n}\right)\right]^{\top}$. We can write $\Psi_n^{\top}\Psi_n$ in the form:
\begin{equation}
    \Psi_n^{\top}\Psi_n =  \Psi_{n-1}^{\top}\Psi_{n-1} + \pmb{\psi}_{n}\pmb{\psi}_{n}^{\top}
\end{equation}
So $\left(\Psi^{\top}_n\Psi_n\right)^{-1}$ can be calculated from $\left(\Psi^{\top}_{n-1}\Psi_{n-1}\right)^{-1}$ and $\pmb{\psi}_n$ by the Sherman-Morrison formula \citep{sherman1950adjustment}.

When $\Psi_n$ has one more column than $\Psi_{n-1}$:
\begin{equation}
    \Psi_n = \begin{bmatrix} \Psi_{n-1} & \pmb{\psi}^{N+1}
    \end{bmatrix}
\end{equation}
We can write $\Psi_n^{\top}\Psi_n$ in the form:
\begin{equation}
    \Psi_n^{\top}\Psi_n =  \begin{bmatrix} 
    \Psi_{n-1}^{\top}\Psi_{n-1} & \Psi_{n-1}^{\top}\pmb{\psi}^{N+1}\\
    \left(\pmb{\psi}^{N+1}\right)^{\top}\Psi_{n-1} & \left(\pmb{\psi}^{N+1}\right)^{\top}\pmb{\psi}^{N+1}
    \end{bmatrix}
\end{equation}
So $\left(\Psi^{\top}_n\Psi_n\right)^{-1}$ can be related to $\left(\Psi^{\top}_{n-1}\Psi_{n-1}\right)^{-1}$ by the block matrix inversion formula \citep{petersen2008matrix}.

The detailed updating rule of the proposed method is given explicitly in Algorithm~\ref{alg:efficient}. The basic structure of this algorithm is identical to Algorithm~\ref{alg:naive}, however the updating rules discussed above are used to avoid recalculating some quantities from scratch. We also establish a recursive relationship between $\hat{\pmb{\theta}}_{n+1}$ and $\hat{\pmb{\theta}}_{n}$. Curiously, the recursive formula has a form very similar to pre-conditioned SGD estimator (with the inverse of Gram matrix as the pre-conditioner). When $n\neq\lfloor (N+1)^{\frac{2\alpha+d}{d}}\rfloor$, the recursion is:
\begin{equation}
\label{eq:olsandsgd}
\hat{\pmb{\theta}}_{n}=\hat{\pmb{\theta}}_{n-1}+\Phi_{n} \pmb{\psi}_{n}\left[Y_{n}-\hat f_{n-1,N}(X_{n})\right]
\end{equation}
Note that for SGD the updating rule replaces $\Phi_n$ by $I$ the identity matrix, i.e. omitting the correlation of $\psi_j$'s w.r.t. the empirical measure. When features are added, there is still a geometrical interpretation, we present the result in Appendix~S3.
\begin{algorithm}[!tb]
\caption{Rule for updating $\hat{\pmb{\theta}}$ with a new observation $(X_n,Y_n)$ efficiently. At step $*$, the value of $\Psi^{\top}_{n-1}\pmb{Y}_{n-1}$ stored in memory needs to be used to avoid repeating calculation.}
\label{alg:efficient}
\begin{algorithmic} 
\STATE {\bfseries Input:} $(X_i)_{i=1}^n,\pmb{Y}_{n},N,\Phi_{n-1},\Psi_{n-1}, a,\Psi^{\top}_{n-1}\pmb{Y}_{n-1}$
\STATE \textbf{function} \text{Updatecurrent} ($X_n, N, \Phi_{n-1},\Psi_{n-1}$) \textbf{output} $(\Phi_{n},\Psi_n)$
\STATE $\quad\pmb{\psi}_n \gets [\psi_1(X_n),\psi_2(X_n),...,\psi_N(X_n)]^{\top}$
\STATE $\quad\Psi_n\gets  [\Psi_{n-1}^{\top}\  \pmb{\psi}_{n}]^{\top}$  , $\Phi_{n} \gets \Phi_{n-1} -\frac{\Phi_{n-1}\ \pmb{\psi_n} \ \pmb{\psi_n}^{T}\  \Phi_{n-1}}{1\ +\ \pmb{\psi_n}^{T}\  \Phi_{n-1}\ \pmb{\psi_n}}$

\STATE \textbf{function} AddBasis ($(X_i)_{i=1}^n, N,\Phi_{n},\Psi_{n}$) \textbf{output} $(\Phi_{n},\Psi_n)$
\STATE $\quad\pmb{\psi}^{N+1} \gets [\psi_{N+1}(X_1),\psi_{N+1}(X_2),...,\psi_{N+1}(X_n)]^{\top}$
\STATE $\quad c\gets\left(\pmb{\psi}^{N+1}\right)^{\top}\pmb{\psi}^{N+1}$ $\quad\pmb{b}\gets \Psi_{n}^{\top} \pmb{\psi}^{N+1}$ $\quad k \gets c - \pmb{b}^{\top} \Phi_{n}\ \pmb{b}$
\STATE $\quad\Psi_n\gets \begin{bmatrix} \Psi_{n} & \pmb{\psi}^{N+1}\end{bmatrix}$ , $\Phi_{n} \gets 
\left[\begin{array}{cc}{\Phi_{n}+\frac{1}{k} \Phi_{n}\ \pmb{b} \pmb{b}^{T}\ \Phi_{n}} & {-\frac{1}{k} \Phi_{n}\ \pmb{b}} \\ {-\frac{1}{k} \pmb{b}^{T} \ \Phi_{n}} & {\frac{1}{k}}\end{array}\right]
$
\vspace{2mm}
\STATE $(\Phi_{n},\Psi_n) \gets \text{UpdateCurrent}(X_n, N, \Phi_{n-1},\Psi_{n-1})$
\IF{$n= \text{Floor}( (N+1)^{2a+1})$}
\STATE $(\Phi_{n},\Psi_n) \gets\ \text{AddBasis}((X_i)_{i=1}^n,N,\Phi_{n},\Psi_{n})$
\STATE $N\ \gets\ N+1$
\ENDIF
\STATE $\hat{\pmb{\theta}}\gets \Phi_n \Psi_n^{\top} \pmb{Y}_n\quad *$
\end{algorithmic} 
\end{algorithm}
\subsection{Computational Expense of Algorithm \ref{alg:efficient}}
\label{efficientsection}
We now show that the computational expense of the updating rule in Algorithm~\ref{alg:efficient} is on average $O(n^{\frac{2d}{2\alpha+d}})$.

When $n\neq \lfloor (N+1)^{\frac{2\alpha+d}{d}}\rfloor$, we do not add a new feature $\psi_{N+1}$ but only update the $\Phi_{n-1}$ matrix with the current $N$ features. The most expensive step is the inner product of $\Phi_{n-1}$ and $\pmb{\psi_n}$, which is an $N\times N$ matrix multiplied by a $N\times 1$ vector. Since the $N = \Theta( n^{\frac{d}{2\alpha+d}})$ at step $n$, the update is of order $n^{\frac{2d}{2\alpha+d}}$.

When $n =  \lfloor (N+1)^{\frac{2\alpha+d}{d}}\rfloor$, we add both a column and a row to the design matrix $\Psi_{n-1}$. The most expensive step is calculating the vector $\pmb{b}$, which gives the pair-wise inner product between $\psi_{N+1}$ and $\left(\psi_j\right)_{j=1}^N$ with respect to empirical measure. In this step an $N\times (n-1)$ matrix is multiplied by an $(n-1)\times 1$ vector, which requires computation of order $n^{\frac{2\alpha+2d}{2\alpha+d}}$. However the algorithm adds new features less frequently as $n$ increases. Thus, in calculating average computational cost, we amortize this expense over all the updates between the inclusion of new basis functions. 

Let
\begin{equation*}
\begin{array}{c}{n=(N)^{\frac{2\alpha+d}{d}}} \\ {n^{+}=(N+1)^{\frac{2\alpha+d}{d}}}
\end{array}
\end{equation*}
that is, $n$ is the first step when there are more than $N$ features included; $n^+$ is the first step when there are more than $N+1$ features. Then the length of the interval between the two "basis adding" steps is
\begin{equation*}
\begin{aligned} n^{+}-n &=(N+1)^{\frac{2\alpha+d}{d}}-(N)^{\frac{2\alpha+d}{d}} \\ & =\Theta( N^{2 \alpha/d}) = \Theta(n^{\frac{2 \alpha}{2 \alpha+d}}) \end{aligned}
\end{equation*}
Thus, $O(n^{\frac{2\alpha+2d}{2\alpha+d}})$ computation is done per $n^{\frac{2\alpha}{2\alpha + d}}$ steps, which is, on average, $O(n^{\frac{2d}{2\alpha+d}})$ per step. Thus the average computational expense of a \emph{single update} with Algorithm \ref{alg:efficient} is of order $n^{\frac{2d}{2\alpha+d}}$. 

\section{Theoretical Analysis of the Online Projection Estimator}
\label{theorysection}
In this section, we formally show that the online estimator in this paper achieves the optimal statistical convergence rate when the true regression function belongs to the hypothesized RKHS. In previous theoretical analysis of (batch) projection estimators \citep{introtononpara}, the proof is shown when $\psi_j$'s are orthogonal to each other w.r.t. the empirical measure of the covariates. This event has probability zero if $X$ has a continuous density. In this section, we show it is possible to get a rate-optimal bound on the generalization error of $\hat f_{n,N}$ even if $\psi_j$'s (the eigenfunctions of the kernel with respect to our ``convenient'' working distribution) are quite correlated w.r.t. the empirical measure of $X$.

Recall that $\mathcal{F}_N = {\text{span}(\psi_1,...,\psi_N)}$ is the linear space spanned by the first $N$ eigenfunctions. Define the \emph{population} minimizer $f_N$ over $\mathcal{F}_N$ as
\begin{equation}
    f_N:= \arg \min_{f\in\mathcal{F}_N} \mathbb{E}[(f(X)-f_{\rho}(X))^2]
\end{equation}
Here we remind our reader that $\hat f_{n,N}\in\mathcal{F}_N$ is the estimator, $f_{N}$ is the population risk minimizer over $\mathcal{F}_N$ and $f_{\rho}\in \mathcal{H}$ is the target function to be estimated. To establish the result that $\|\hat f_{n,N}-f_{\rho}\|_2\rightarrow 0$ as $n\rightarrow\infty$, we first bound the rate at which $\|\hat f_{n,N}-f_N\|_2$ goes to $0$ as $N$ grows (sufficiently slowly); then we bound the rate at which $\|f_N-f_{\rho}\|_2\rightarrow 0$ as $N\rightarrow\infty$. With the correct choice of $N= \Theta(n^{\frac{d}{2\alpha+d}})$ we can balance the rate of the above two term converging to $0$. Before we state the result, we give assumptions necessary for the proof.

(A1) The joint distribution of i.i.d. $(X_i,Y_i)$ has support $\mathcal{X}\times \mathbb{R}\subset \mathbb{R}^d\times \mathbb{R}$ and $\mathcal{X}$ is compact. The i.i.d. zero-mean noise random variables $\epsilon_i = Y_i - f_{\rho}(X_i)$ satisfy the following
\begin{equation}
    \|\epsilon_i\|_{m, 1} := \int_{0}^{\infty} \mathbb{P}(|\epsilon_i|>t)^{1 / m} dt < \infty\quad\text{for some }m>1
\end{equation}
\textit{Note.} If for some $\delta > 0, m>1$ we have that $m+\delta$ moment of $\epsilon_i$ exists, then (A1) is satisfied for that value of $m$. This is \emph{slightly} stronger than existence of the $m$-th moment, cf. \cite[Chapter~10]{ledoux2013probability}

Our noise assumption is substantially weaker than the typical sub-Gaussian noise assumptions (sub-Gaussian random variables have all moments bounded). In the light tail noise setting, the level of the noise only influences the convergence speed by at most a constant. However we will see in Theorem~\ref{maintheorem}, if the eigenvalues decrease too fast (the RKHS is too small) and the noise has too few moments, the convergence \emph{rate} will depend on the noise level. Our analysis characterizes the interplay between the size of RKHS space and the noise level, using a sharp multiplier inequality \cite[Theorem~1]{han2019convergence}. There are currently no other methodologies, to our knowledge, that are both computationally tractable and have provable convergence guarantees with heavy-tailed noise in the online non-parametric regression setting.


(A2) The true regression function $f_{\rho}$ belongs to the known RKHS $\mathcal{H}$ i.e. the RKHS-norm $\|f_{\rho}\|_{\mathcal{H}}$ is finite.




(A3) The kernel function has Mercer expansion $K(x,z) = \sum_{j=1}^{\infty} \lambda_j \psi_j(x)\psi_j(z)$, where $\left(\psi_j\right)_{j=1}^{\infty}$ are orthonormal with respect to some specified \emph{working distribution} $\bar\rho_X$, and $\lambda_j = \Theta(j^{-2\alpha/d})$ with $\alpha>d/2$.

(A4) The distribution of $X$, $\rho_X$, is absolutely continuous w.r.t. $\bar\rho_X$. Let $p_X = d\rho_X/d\bar\rho_X$ denote its Radon–Nikodym derivative. We assume for some $D<\infty$:
    \begin{equation*}
        p_X(x) \leq D \quad\text{for all }x\in\mathcal{X}
    \end{equation*}
\textit{Note.} In the (very common) case that both of these have densities with respect to Lebesgue measure, this is equivalent the ratio of their densities being bounded.   

\begin{theorem}[Optimal convergence rate]   
\label{maintheorem}
Assume (A1-A4), let $\hat f_{n,N}$ be the projection estimator \eqref{PE}. Assume that $\|\hat f_{n,N}\|_{\infty}\leq M$, for some $M<\infty$. Choosing $N = \Theta( n^{\frac{d}{2\alpha + d}})$, we have
\begin{equation}
    \|\hat f_{n,N} - f_{\rho}\|_2 = O_P\left(n^{-\frac{\alpha}{2\alpha+d}} \sqrt{\log n} \vee n^{-\frac{1}{2}+\frac{1}{2m}} \sqrt{\log n}\right)
\end{equation}
If $m\geq 2$ in (A1), the above bound holds in expectation:
\begin{equation}
\label{mainresult}
    \mathbb{E}[\|\hat f_{n,N} - f_{\rho}\|_2] = O\left(n^{-\frac{\alpha}{2\alpha+d}} \sqrt{\log n} \vee n^{-\frac{1}{2}+\frac{1}{2m}} \sqrt{\log n}\right)
\end{equation}
\end{theorem}
Note as long as all the moments of $\epsilon_i$'s exist (e.g. when $\epsilon_i$'s are sub-exponential), the convergence rate only depends on the size of the RKHS. One merit of our method is that even if the noise does not have finite variance, that is, $m< 2$ in (A1), our method still has convergence guarantees. To our knowledge, existing work on non-parametric SGD does not give convergence guarantees with such heavy tailed noise.
 

As we compare the two components on the RHS of the bound presented in \eqref{mainresult}, we can see that when $m>\frac{2\alpha}{d}+1$, that is, when we have a relative light-tailed noise, our bound is dominated by the size of the RKHS. However when $m<\frac{2\alpha}{d}+1$, it is the noise that dominates our bound. Also note that as $d$ increases, fewer moments on $\epsilon$ are required for our bound to match the classical non-parametric minimax rate in our RKHS.

The following lower bound demonstrates that this rate of convergence is indeed optimal (up to a logarithm term) among all estimators. For $\lambda_j = \Theta(j^{-2\zeta})$ (to compare with Theorem~\ref{maintheorem}, take $\zeta = \alpha/d$), let $B_R = \{f\in\mathcal{H}\ |\ \|f\|_{\mathcal{H}}\leq R\}$ be the $R$-ball in RKHS $\mathcal{H}$. Then we have the minimax bound:
\begin{equation}
    \liminf_{n\rightarrow\infty} \inf_{\hat f} \sup_{f_{\rho}\in B_R} \mathbb{E}\left[n^{\frac{\zeta}{2\zeta+1}}\|\hat f-f_{\rho}\|_2\right] \geq C
\end{equation}
where the infimum ranges over all possible functions $\hat f$ that are measurable of the data. For a derivation of the lower bound, see \citet[Chap.~15]{wainwright2019high}.

Upper bounds similar to our results in Theorem~\ref{maintheorem}  have been shown in \cite{tarres2014online,dieuleveut2016nonparametric}, for SGD-type nonparametric online methods. However the proposed estimators there use $n$ basis function, therefore have an unacceptable $\Theta(n^2)$ total computational expense. There are methods that aim to improve the computational aspect by using random features or other acceleration methods (cf. Section~\ref{section:citation}), however the theoretical guarantees on statistical convergence rates in that work are generally quite weak (generally giving upper bounds of $n^{-1/4}$ in RMSE, which is far from the minimax rate) and insensitive to the decay rate of eigenvalues. 


Many existing online nonparametric estimators aim to find a function $f\in\mathcal{F}$ that minimizes an expected convex loss $\mathbb{E}[l(f(X),Y)]$, which is a more general setting than this study. However, the majority of previous work on this topic assumes the loss function $l(\cdot,\cdot)$ is Lipschitz w.r.t. the first argument, including \citet{dai2014scalable, si2018nonlinear, koppel2019parsimonious, lu2016large}. Specializing to the regression problem (with squared-error-loss), this is essentially assuming the outcomes $Y_i$ (therefore the noise $\epsilon_i$) are uniformly bounded: because $l(f(x),y)-l(f(z),y) = (f(x)-y)^2 - (f(z)-y)^2 = (f(x)-f(z))(f(x)+f(z)-2y)$. If we require $l(\cdot,\cdot)$ to be Lipschitz, we basically require $f(x),f(z), y$ to be uniformly bounded. Although we still only consider bounded $f$ in this work, we relax the contraint on the noise variables: we require only finite moments of $\epsilon_i$ and show (in)sensitivity of our bound.

\section{Multivariate Regression Problems}
In most applications, the covariate $X_i$'s take value in $\mathbb{R}^d$ where $d>1$. If the kernel function $K:\mathbb{R}^d\times\mathbb{R}^d\rightarrow\mathbb{R}$ has a known Mercer expansion \eqref{kernelexpension}, then the proposed method can be applied directly. If the kernel function takes a tensor product form (e.g. the Gaussian kernel), or is constructed from a 1-dimension kernel via a tensor product (e.g. $K(x,z) = \prod_{k=1}^d\min\{x^{(k)},z^{(k)}\}$, where $x^{(k)}$ is the $k$-th entry of $x\in\mathbb{R}^d$), the eigenvalues and eigenfunctions are just the tensor product of the 1-dimensional kernels' \cite[Section~3.5]{michel2012lectures}, \cite[Section~5.2]{xiu2010numerical}. However, as presented in Section~\ref{theorysection}, the minimax rate of estimating in a $d$-dimension $\alpha$-order Sobolev space is $\Theta(n^{-\frac{\alpha}{2\alpha+d}})$, which becomes quite slow when $d$ is large (unless at the same time a large $\alpha$ is assumed). 

A popular low-dimension structure people have used is the nonparametric additive model \citep{hastie2009elements,yuan2016minimax}, which is thought to effectively balance model flexibility and interpretability. For $x\in\mathbb{R}^d$, we might consider imposing an additive structure on our model~\eqref{reg}:
\begin{equation}
    f_{\rho}(x) = \sum_{k=1}^d f_{\rho,k}\left(x^{(k)}\right)
\end{equation}
where the component functions $f_{\rho,k}$ belong to a RKHS $\mathcal{H}$ (in general they can belong to different spaces). For a fixed $d$, the minimax rate for estimating an additive model is identical (up to a multiplicative constant $d$) to the minimax rate in the analogous one-dimension nonparametric regression problem that works with the same hypothesis space $\mathcal{H}$ \citep{raskutti2009lower}. The proposed online method can be directly generalized to this setting, for more discussion and empirical performance, see Appendix S4.

\section{Simulation Study}
In this section, we illustrate both the computational and statistical efficiency of the online projection estimator, in both one-dimension regression and additive model settings. 
\subsection{Generalization Error of the Online Projection Estimator is Rate-Optimal}
In this section, we use simulated data to illustrate that the generalization error of our estimator reaches the minimix-optimal rate. For each sample, $X_i$ is generated from $\rho_X$ whose density function is $p_X(x)$; $Y_i$ is generated by $Y_i=f_{\rho}(X_i)+\epsilon_i$. The details of the parameters are listed in Table \ref{settingtable}. In example 1, we purposely select $\rho_X$ such that $\int_0^1 \psi_i(x)\psi_j(x)p_X(x)dx = \delta_{ij}$, together with bounded noise. In example 2, basis functions are no longer orthogonal w.r.t. $\rho_X$ and a low signal-noise ratio is applied. In both simple and more realistic scenarios, the online projection estimator achieves rate-optimal statistical convergence.

The $f_{\rho}$ in example 1 is taken from \citet{dieuleveut2016nonparametric}, where they used it to illustrate the performance of the functional SGD estimator; the regression function in example 2 is also used in a study of wavelet neural networks \citep{alexandridis2013wavelet}.

In example 1, the hypothesis space is the second-order spline on the circle
\begin{equation*}
\begin{aligned} W_{2}^0(per) &=\left\{f \in L^{2}([0,1])\ |\ \int_{0}^{1} f(u) d u=0\right.\\ &f(0) \left.=f(1), f^{\prime}(0)=f^{\prime}(1), \int_{0}^{1}\left(f''(u)\right)^{2} d u<\infty\right\} \end{aligned}
\end{equation*}
In example 2, we use Sobolev space $W_1^0([0,1])$ defined in \eqref{eq:W10}. Because eigenvalues decrease faster in example 1, we observe a convergence rate of $\sim n^{-4/5}$, which is faster that that in example 2, $\sim n^{-2/3}$. 

\begin{table}[t!] 
\caption{Settings of simulation studies. $^*B_4(x) = x^{4}-2 x^{3}+x^{2}-\frac{1}{30}$ is the 4-th Bernoulli polynomial, and $\{x\}$ means taking the fractional part of $x$. }
\label{settingtable}\par
\vskip .2cm
\centerline{\tabcolsep=3truept}
\begin{tabular}{ccc} \hline 
& Example 1 & Example 2 \\
\hline
Kernel $K(s,t)$ & $\frac{-1}{24}B_4(\{s-t\})^*$& $\min\{s,t\}$ \\
Eigenvalue $\lambda_j$ & $\frac{2}{(2 \pi j)^{4}} = O(j^{-4})$&$\frac{4}{(2j-1)^2\pi^2} = O(j^{-2})$ \\
Basis function $\psi_j(x)$ & $\sin(2\pi j x),\cos(2\pi j x)$& $\sqrt{2}\sin(\frac{(2j-1)\pi x}{2})$\\
$p_{X}(x)$ &$\bm{1}_{[0,1]}(x)$&$(x+0.5)\bm{1}_{[0,1]}(x)$\\
Noise $\epsilon$ & Unif([-0.02,0.02])& Normal(0,5)\\
\multirow{2}{*}{True regression function $f_{\rho}$} & \multirow{2}{*}{$B_4(x)$}& $(6x-3)\sin(12x-6)$\\
&$+\cos^2(12x-6)$\\
\hline
\end{tabular}
\end{table}

We use $\|\hat f_{n,N}-f_{\rho}\|_2^2$ as a measure of goodness of fitting (Figure \ref{MSEplot}). The method in this paper is compared with an online nonparametric SGD estimator \citep{dieuleveut2016nonparametric} and the kernel ridge regression (KRR) estimator \eqref{KRR}. Although KRR might have a better generalization capacity (the rates should be the same, but there might be an improvement in the constant), it is computationally prohibitive to apply it in the online learning setting, so we only include this method as a reference. The hyperparameters for each method are chosen to optimize performance (oracle hyperparameters). For our method, it is the constant in front of timing for adding new basis functions. In Figure~\ref{curveplot}, we present several typical realizations of $\hat f_{n,N}$ for both examples, together with data points. 
\subsection{CPU Time}
\begin{figure}[t]
\vskip 0.2in
\begin{center}
\centerline{\includegraphics[width=\columnwidth]{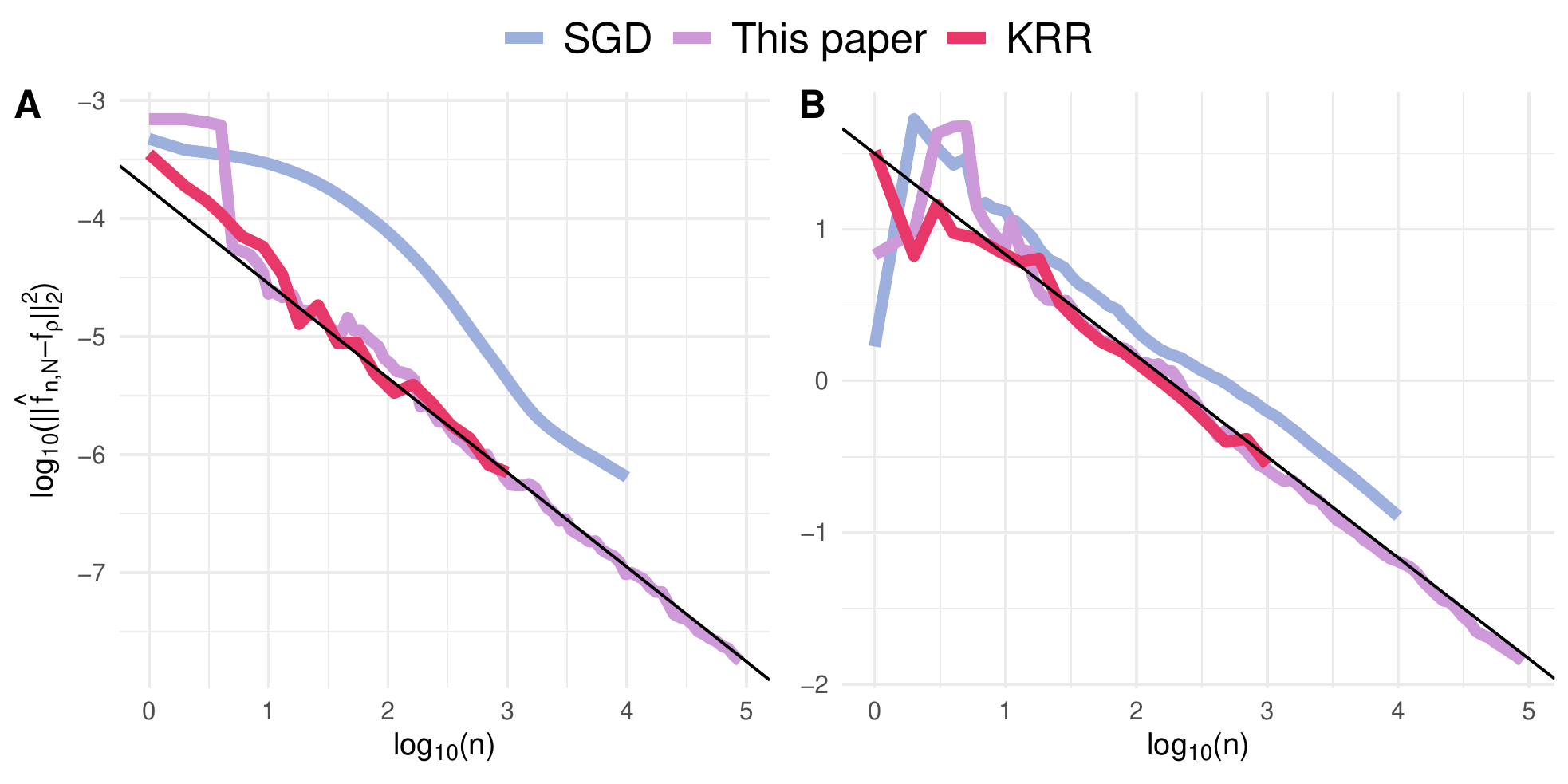}}
\caption{$\log_{10}\|\hat f_{n,N}-f_{\rho}\|_2^2$ against $\log_{10} n$.(\textbf{A}) Example 1, black line has slope $=-4/5$; (\textbf{B}) Example 2, black line has slope $=-2/3$. Each curve is calculated as the average of 15 repetitions. Due to different computational costs, we chose different maximum $n$ for different methods.}
\label{MSEplot}
\end{center}
\vspace{-0.58cm}
\end{figure}
\begin{figure}[!htbp]
\begin{center}
\centerline{\includegraphics[width=\columnwidth]{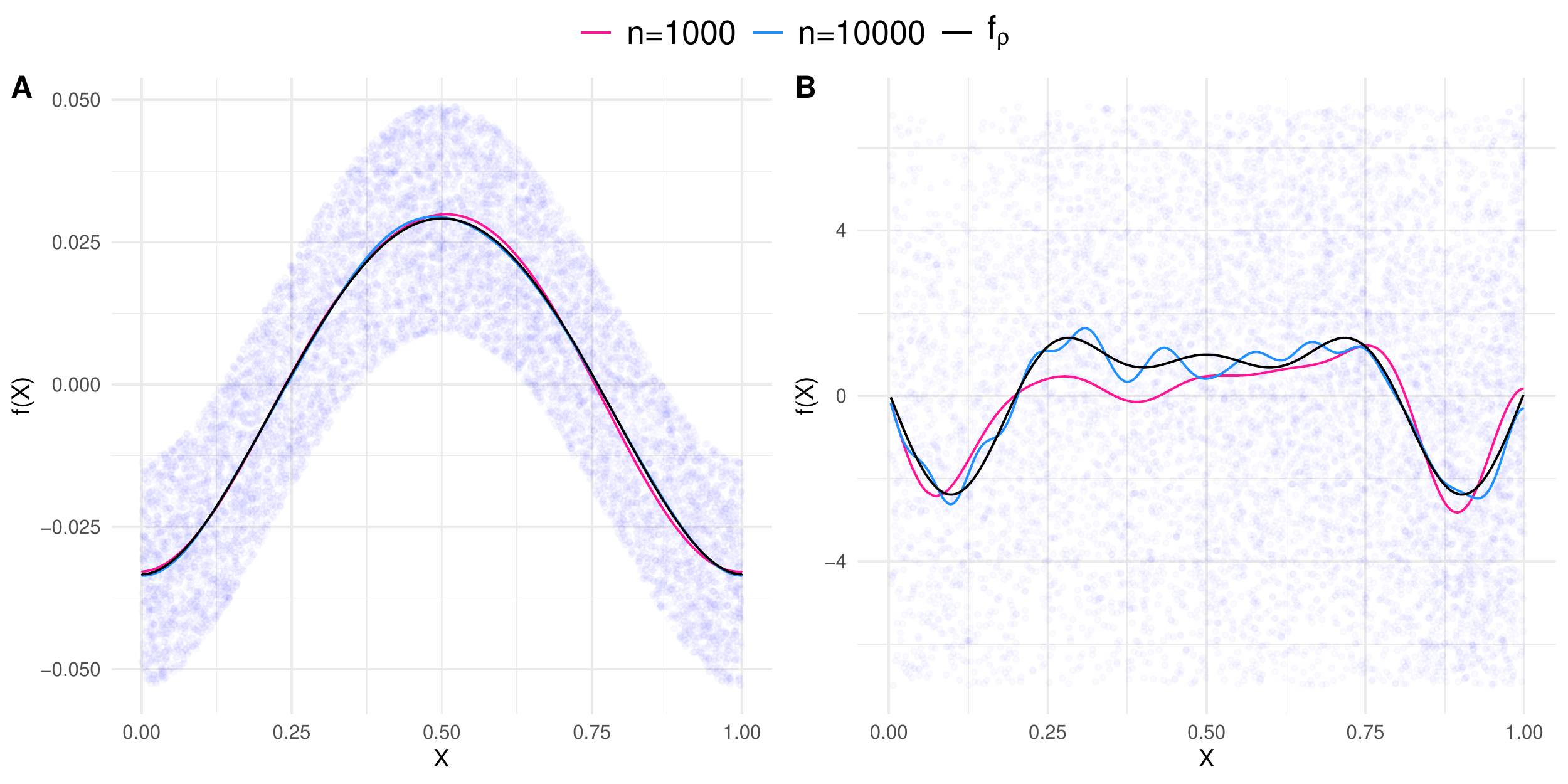}}
\caption{Realizations of $\hat f_{n,N}$. (\textbf{A}) Example 1; (\textbf{B}) Example 2.}
\label{curveplot}
\end{center}
\vskip -0.2in
\end{figure}

Figure~\ref{CPUtime} shows the CPU time used in calculating online estimators for up to $n$ samples when solving example 2, for the online projection estimator and nonparametric SGD estimator. Experiments were run on a computer with 1 Intel Core m3 processor, 1.2 GHz, with 8 GB of RAM. For the projection estimator, new basis functions are added when $n=\lfloor N^{2\alpha+1}\rfloor,$ $N=1,2,...$. First, we can see, for all $\alpha \in \{1,2,3\}$, the online projecting estimators are all significantly faster to compute than nonparametric SGD estimator after $n>10^4$, because the latter requires evaluation of $n$ basis functions for the $n+1$st update, which will accumulate very fast. In addition, for larger $\alpha$ the total computational cost for the online projection estimator becomes nearly linear in $n$. There are also some ``jumps'' in the CPU time for the online projection estimator: They correspond to steps when new basis functions are added in. Both of the phenomena match our analysis in Section~\ref{efficientsection}. Although it seems beneficial both computationally and statistically to use a larger $\alpha$, it is important to remember that $\alpha$ too large may result in poor generalization error -- this occurs if the RKHS associated with $\alpha$ becomes so small that it no longer includes $f_{\rho}$ (see discussion in \citet{noahmiss}). 

\begin{figure}[!ht]
\vskip 0.2in
\begin{center}
\centerline{\includegraphics[width=0.8\columnwidth]{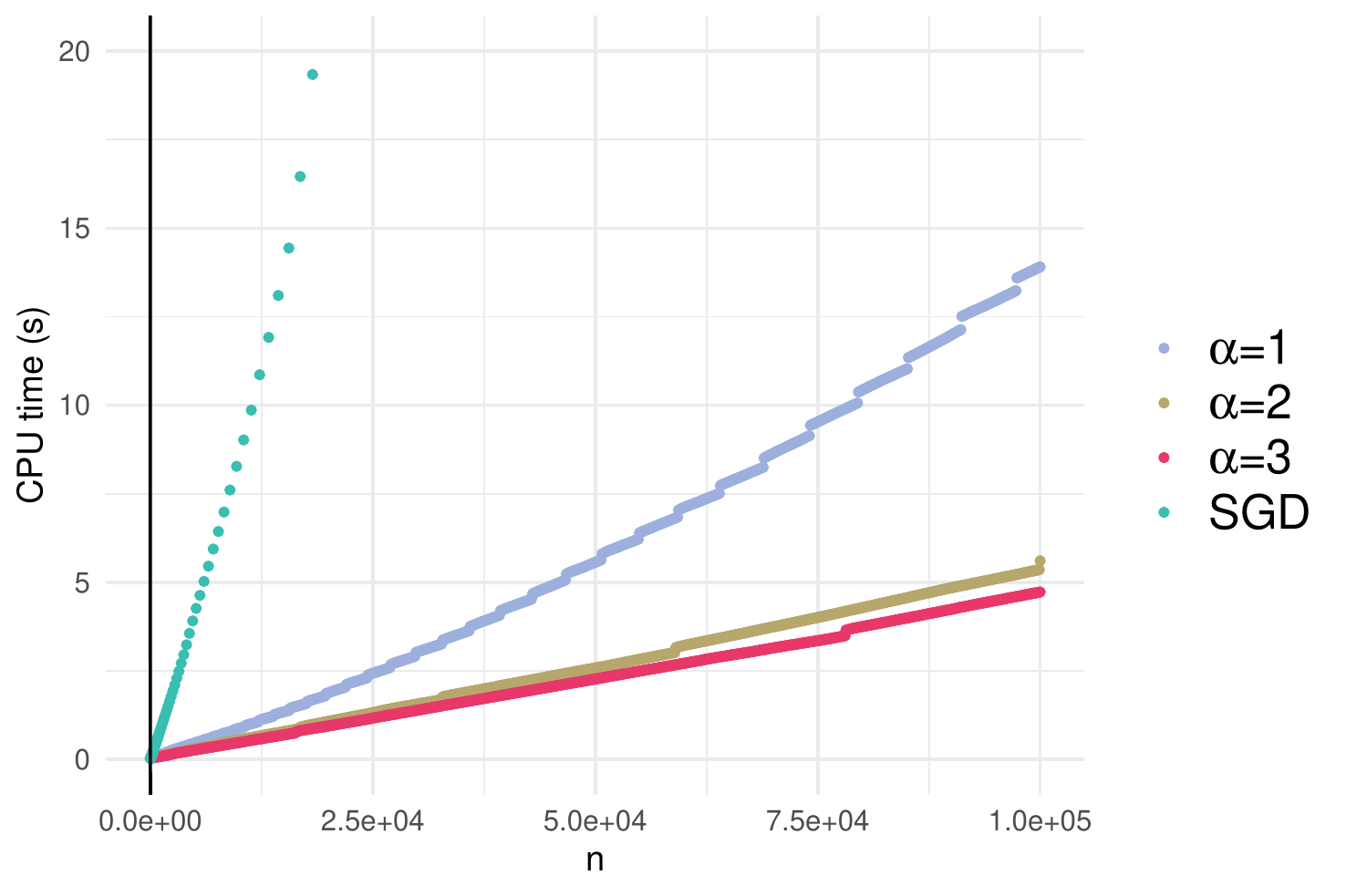}}
\caption{CPU time against sample size (10 runs each curve).}
\label{CPUtime}
\end{center}
\vskip -0.2in
\end{figure}

\section{Discussion}
In this paper, we proposed a framework to construct online nonparametric regression estimators when the hypothesis space is a RKHS. We showed that: (i) the error of the proposed estimator is near-optimal; and (ii) the computational expense of calculating such estimators is much lower than other contemporary estimators with similar statistical guarantees. In addition, our estimator is actually precisely an empirical risk minimizer (in a linear space of slowly growing dimension), which allows us to give theoretical guarantees when the noise is heavy tailed (as compared to the previously required assumptions of boundedness).


In this work, we leveraged properties of least-squares loss to efficiently update the empirical risk minimizer $\hat f_{n,N}$ in an online manner. However, for a general convex loss function (e.g. logistic regression), construction of an online nonparametric estimator that has both guaranteed optimal generalization capacity and is computationally feasible for larger problems is still an open question. Although there are functional SGD type estimators designed for this purpose (discussion in Section~\ref{section:citation}), it would be interesting to see if it is possible to design estimators that are both computationally efficient to update and are (approximate) ERM in a deterministic space.

\vskip 14pt
\noindent {\large\bf Supplementary Materials}\\
In the Appendix, we provide provide proof of Theorem~\ref{maintheorem}. In the later sections, we give a complete description of settings for simulations from the main text, together with more examples. We also include some additional discussion on the applications of our estimator.
\par
\vskip 14pt
\noindent {\large\bf Acknowledgements}

N.S and T.Z. were both supported by NIH grant R01HL137808.

\newpage
\begin{appendices}
\section{Supplementary Discussion on RKHS}
In the main text we gave two equivalent definitions of RKHS: one based on the reproducing property and another one based on the Mercer expansion of the kernel.

The proposed method directly works with the eigenfunctions $\psi_j$, and it does not directly approximate either the kernel function $K$ or the kernel matrix $\mathbb{K}$. Although in many cases we start with a Mercer kernel in hand and calculate its eigendecomposition afterwards, it is not uncommon to begin with features and then attempt to calculate a closed-form of an implied kernel. This situation suits perfectly with our method: for the well-known the smoothing spline method proposed in \citet[Chapter~2]{wahba1990spline}, the author starts with $\psi_j(x) = \sin(2j\pi x),\cos(2j\pi x)$ and shows us how to get the closed-form of the reproducing kernel for periodic Sobolev space $W_m^0(\text{per})$. However, such a Bernoulli polynomial closed-form of the kernel is no longer available when $m$ is not an integer, which corresponds to a fractional Sobolev space case; when considering kernel space on sphere $\mathbb{S}^2$, some effort is required to obtain the closed-form expression even for simple cases (\cite{kennedy2013classification}, \cite{michel2012lectures}), but the features are just orthonormal spherical harmonics; for multiscale kernels defined by compactly-supported wavelet eigenfunctions \citep{opfer2006multiscale} or Legendre polynomials \cite[Section~3.3.2]{xiu2010numerical}, it is also simplest to work directly with features rather than attempting to identify a closed-form expression for the implied kernel.  

In the main text we provide the Mercer expansion of a Sobolev space $W_1^0([0,1])$. We also state the (correct) expansion for Gaussian kernel (there are several versions in the literature that are not correctly normalized):

When $\bar{\rho}_X$ has density (w.r.t Lebesgue measure on $\mathbb{R}$) $\bar{p}_X = \frac{\alpha}{\sqrt{\pi}}\exp(-\alpha^2x^2)$, we have the expansion of Gaussian kernel $K(x,z) = \exp(-\epsilon^2|x-z|^2)$ with
\begin{equation}
\label{eq:gaussianeigen}
    \begin{aligned}
        \lambda_j = \sqrt{\frac{\alpha^2}{\alpha^2 + \delta^2 + \epsilon^2}}\left(\frac{\epsilon^2}{\alpha^2+\delta^2+\epsilon^2}\right)^{j-1}\\
        \psi_j(x)= \gamma_j \exp(-\delta^2x^2)H_{j-1}(\alpha\beta x)
    \end{aligned}
\end{equation}
where the $H_j$ are Hermite polynomials of degree $j$, and
\begin{equation}
    \beta = \left(1+\left(\frac{2\epsilon}{\alpha}\right)^2\right)^{1/4},
    \ \gamma_j = \sqrt{\frac{\beta}{2^{j-1}\Gamma(j)}},
    \ \delta^2 = \frac{\alpha^2}{2}\left(\beta^2 - 1\right)
\end{equation}

The multivariate Gaussian kernel's eigenfunctions and eigenvalues are just the tensor product of the 1-dimension Gaussian kernel. Formally, the multivariate Gaussian kernel $K(\bold{x},\bold{z}) = \exp(-\epsilon^2\|\bold{x}-\bold{z}\|^2)$ has the following expansion:
\begin{equation}
    K(\bold{x},\bold{z}) = \sum_{\bold{j}\in \mathbb{N}^d} \lambda^*_{\bold{j}}\psi^*_{\bold{j}}(\bold{x})\psi^*_{\bold{j}}(\bold{z})
\end{equation}
where the eigenvalues and eigenfunctions are related to \eqref{eq:gaussianeigen} as
\begin{equation}
    \lambda_{\bold{j}}^* = \prod_{l = 1}^d \lambda_{\bold{j}_l},\  \psi_{\bold{j}}^*(\bold{x}) = \prod_{l = 1}^d \psi_{\bold{j}_l}(x^{(l)}), 
\end{equation}
where $x^{(l)}$ is the $l$-th component of $x\in\mathbb{R}^d$. There are also available numerical methods (independent of $(X_i,Y_i)'s$) for approximating kernel eigenfunctions in cases where analytical forms are not available, see  \cite{rakotch1975numerical,santin2016approximation}, \cite[Section~4.3]{rasmussen2003gaussian}, \cite{cai2020eigenvalue} and \cite[Chapter~12]{fasshauer2015kernel}.

There is also an interesting formal similarity between Mercer expansions and Bonchner's theorem (see, e.g. \cite{rahimi2007random}) which gives rise to random Fourier feature-based methods. On one hand, we have the Mercer expansion:
\begin{equation}
    K(x,z) = \sum_{j=1}^{\infty} \lambda(j) \psi(x,j)\psi(z,j)
\end{equation}
On the other hand, the positive-definite (real-valued) kernel has a convolutional representation by Bonchner's theorem \citep{rahimi2007random}):
\begin{equation}
\label{equation:RFFexpansion}
    K(x,z) =\int_{\mathcal{X}\times [0,2\pi]} p(\omega, b) \cos(\omega^{\top} x + b)\cos\left({\omega^{\top} z + b}\right) d\omega db
\end{equation}
The random Fourier feature expansion \eqref{equation:RFFexpansion} uses a set of basis functions (cosines) that is not sensitive to the expanded kernel. Only the probability distribution we sample $\omega$ from depends on the kernel. Such a choice may bring some convenience in application, but at the price of using an approximation that converges to the kernel much slower. Another difference is in the basis selection strategy: For the Mercer expansion it is very straightforward -- we choose the eigenfunctions corresponding to larger eigenvalues. By this strategy, we can ensure the features we choose are more important and orthogonal to each other w.r.t. RKHS inner product. For random feature-based methodologies, one has to sample from a probability distribution because there are uncountably infinitely many $\omega$ (versus countably infinite $j$) and there is less we can say about the geometric properties of random features \citep{yu2016orthogonal}.

Our readers can also find expansions of various kernels in \cite{wainwright2019high,wahba1990spline, fasshauer2012green,williams2000effect,shi2009data, liang2014eigen, fornberg2008stable}. There are also several existing online nonparametric learning methods not mentioned in the main text, e.g. \cite{kivinen2001online, ying2006online, rudi2017generalization,alaoui2015fast, xiong2019online} .
\section{Proof of Theorem 3}
We can decompose the $L^2_{\rho_X}$-distance(i.e. $\|\cdot\|_2$-distance) between $\hat f_{n,N}$ and $f_{\rho}$ into two parts by inserting a $f_N$ function in between. Recall the definition of the previous two are:
\begin{equation}
    \begin{aligned}
        \hat f_n &:= \underset{f \in \mathcal{F}_N}{\operatorname{argmin}} \frac{1}{n} \sum_{i=1}^{n}\left(Y_{i}-f\left(X_{i}\right)\right)^{2}\\
        f_{\rho} & :=\underset{f \in L^2_{\rho_X}}{\operatorname{argmin}} \int_{\mathcal{X}\times \mathbb{R}}\left(Y-f\left(X\right)\right)^{2}d\rho(X,Y)
    \end{aligned}
\end{equation}
where $\mathcal{F}_N$ is a subset of the $N$-dimension vector space spanned by $\psi_1,...,\psi_N$:
\begin{equation}
\label{eq:defineFN}
    \mathcal{F}_N = \mathcal{F}_N(M):=\{f\in L^2_{\rho_X}\mid f\in \text{span}(\psi_1,...,\psi_N), \|f\|_{\infty}<M\}
\end{equation}.
We insert a deterministic function $f_N$ in-between to facilitate the use of the triangle inequality. 
\begin{equation}
    f_N :=\underset{f \in \mathcal{F}_N}{\operatorname{argmin}} \int_{\mathcal{X}\times \mathbb{R}}\left(Y-f\left(X\right)\right)^{2}d\rho(X,Y)
\end{equation}
So we have the following decomposition of $L^2_{\rho_X}$ distance:
\begin{equation}
\label{eq:tri}
    E\|\hat f_{n,N} - f_{\rho}\|_2\leq E\|\hat f_{n,N} - f_N\|_2 + \|\hat f_{N} - f_{\rho}\|_2
\end{equation}
If we can bound the two terms at the correct rates separately at the desired order, combining them together would give the result in Theorem~3.

\subsection{Bound $\|f_N - f_{\rho}\|_2$}
We first handle the second term in \eqref{eq:tri}. It is a deterministic quantity which represents the approximation error of our estimator. In the main text, we given two equivalent definitions of RKHS, respectively based on the reproducing property and the Mercer expansion. We will use the second one to explicitly calculate the approximation error. Let $\mathcal{H}$ denote the native space of $K$ (the RKHS of interest).

\begin{lemma}
\label{lemma:lemmabias}
Assume (A1),(A2),(A4), we have
\begin{equation}
    \left\|f_{N}-f_{\rho}\right\|_2 \leqslant \left(D \|f_{\rho}\|_{\mathcal{H}} \lambda_{N}\right)^{1 / 2}
\end{equation}
where $\|\cdot\|_{\mathcal{H}}$ is the RKHS-norm. If we further assume (A3) and choose $N = \Theta(n^{\frac{d}{2\alpha + d}})$, then
\begin{equation}
    \|f_N - f_{\rho}\|_2 = O(n^{-\frac{\alpha}{2\alpha+d}})
\end{equation}
\end{lemma}

\begin{proof} Since $f\in\mathcal{H}$ by assumption, we know $f_{\rho}$ has the following expansion w.r.t $\psi_j$: $f_{\rho} = \sum_{j=1}^{\infty}\theta_j\psi_j$. Recall that we defined $(\lambda_j,\psi_j)$ as the eigen-system of operator $T_{k,\bar{\rho}_X}$ in Section~2. By the definition of RKHS in Proposition~2, the condition $\|f_{\rho}\|_{\mathcal{H}}<\infty $ in (A2) can be rewritten as:
\begin{equation}
    \|f_{\rho}\|_{\mathcal{H}}^2 = \sum_{j=1}^{\infty} \left(\frac{\theta_j}{\sqrt{\lambda_j}}\right)^2 <\infty
\end{equation}
Define $f_{\rho,N} = \sum_{j=1}^N \theta_j \psi_j\in\mathcal{F}_N$ to be a truncated approximation of $f_{\rho}$ (which does not depend on data). We know that $\|f_N-f_{\rho}\|_2$ is smaller than $\|f_{\rho,N}-f_{\rho}\|_2$ because $f_N$ is the minimizer of $\|f-f_{\rho}\|_2$ over $f\in\mathcal{F}_N$.

So we have:
\begin{equation}
\begin{aligned}
      \|f_N - f_{\rho}\|_2 & \leq  \|f_{\rho,N} - f_{\rho}\|_2\\
    & = \left(\int_{\mathcal{X}} \left(f_{\rho,N}(x) - f_{\rho}(x)\right)^2d\rho_X(x) \right)^{1/2}\\
    & \stackrel{(1)}{\leq}  D^{1/2} \left(\int_{\mathcal{X}} (f_{\rho, N}(x) - f_{\rho}(x))^2d\bar\rho_X(x) \right)^{1/2} \\ 
    & \stackrel{(2)}{=}  \left(D\sum_{j=N+1}^{\infty}\theta_j^2\right)^{1/2}\\
    & \leq \left(D\lambda_N \sum_{j=N+1}^{\infty} \theta_j^2\lambda_j^{-1}\right)^{1/2}\\
    & \leq  (D\|f_{\rho}\|_{\mathcal{H}}\lambda_N)^{1/2}
\end{aligned}
\end{equation}
In (1) we use assumption (A4) about the relationship between $\rho_X$ and $\bar{\rho}_X$. In (2) we use Parseval's identity noting that $\psi_j$'s are orthonormal w.r.t. $\bar{\rho}_X$.

If we take $N = \Theta(n^{\frac{1}{2\alpha+d}})$ and assume $\lambda_j =\Theta(j^{-2\alpha/d})$, we have $\lambda_N = \Theta(n^{-\frac{2\alpha}{2\alpha+d}})$, therefore $\|f_N - f_{\rho}\|_2 = O(n^{-\frac{\alpha}{2\alpha+d}})$. Thus we have proven the first part of the Lemma.
\end{proof}
\subsection{Bound $\mathbb{E}\|\hat f_{n,N} - f_N\|_2$}
In this section we bound the term associated with the stochastic error. Our proof engages the following steps: We first show the hypothesis space is a VC-class, then use this property to bound its localized Rademacher complexity. This will further lead us to the final convergence rate because $\hat f_{n,N}$ is an M-estimator (ERM of the negative loss) over this hypothesis space. We use the novel result presented in \cite{han2019convergence} to bound the multiplier process with a Rademacher process, which allows us to quantify the interplay between hypothesis space size and the level of noise.
\begin{proposition}
\label{prop:vclinear}
Let $\mathcal{F}_N$ be the $N$-dimension linear space defined in \eqref{eq:defineFN}, then we know $\mathcal{F}_N$ is VC-subgraph class with index less than or equal to $N+2$.
\end{proposition}
\begin{proof}
The definition of VC-subgraph class, together with the fact that a $N$-dimension vector space $\mathcal{F}_N$ of measurable functions is a VC-class of index no more than $N+2$, can be found in \cite[Lemma ~2.6.15]{van1996weak} or \cite[Proposition~4.20]{wainwright2019high}.
\end{proof}

Now we use the fact that $\mathcal{F}_N$ is a VC-class to get an upper bound on its covering number. For this, we need the following result.

 \begin{proposition}
 \label{prop:universality}
 For a VC-subgraph class of functions $\mathcal{F}$. One has for any probability measure $Q$:
 \begin{equation}
     \mathcal{N}(\epsilon \|F\|_{Q,2},\mathcal{F}, L^2_{Q}) \leq CN(16e)^{N}\left(\frac{1}{\epsilon}\right)^{2(N-1)}
 \end{equation}
 where $N$ is the VC-dimension of $\mathcal{F}$ and $0<\epsilon<1$. And $F$ is the envelope function of $\mathcal{F}$, i.e. $|f(x)|\leq F(x)$ for any $x\in\mathcal{X},f\in\mathcal{F}$.
 \end{proposition}
 \begin{proof}
One can find the proof of a slightly more general version in \cite[Theorem~2.6.7]{van1996weak}.
\end{proof}

For a function space $\mathcal{F}$, define the localized uniform entropy integral as:
\begin{equation}
J(\delta, \mathcal{F}, L_{2}) := \int_{0}^{\delta} \sup_{Q} \sqrt{1+\log \mathcal{N}\left(\epsilon\|F\|_{Q,2}, \mathcal{F}, L_{2}(Q)\right)} d\epsilon
\end{equation}
Applying this to the space $\mathcal{F}_N$, we have the following result:
\begin{lemma}
\label{lemma:entropybound}
Let $\mathcal{F}_N$ be the function space defined in \eqref{eq:defineFN}, we have 
\begin{equation}
    J(\delta,\mathcal{F}_N,L_2) \leq C_M\sqrt{N\delta^2\log\left(\frac{1}{\delta}\right)}
\end{equation}
for sufficiently small $\delta$. The constant $C_M$ only depends on $M$.
\end{lemma}
\begin{proof}
We first note $\mathcal{F}_N$ is a subset of an $N$-dimension vector space with envelope $F(x) = M$. By Proposition~\ref{prop:vclinear} and Proposition~\ref{prop:universality}, we have 
\begin{equation}
\begin{aligned}
    \mathcal{N}(\epsilon M, \mathcal{F}_N, L^2(Q))&\leq CN(16e)^N\left(\frac{1}{M\epsilon}\right)^{2N-2}\quad\text{for any measure }Q\\
    \Rightarrow J(\delta, \mathcal{F}, L^2) &\leq C \int_0^{\delta} \sqrt{N\log\left(\frac{1}{M\epsilon}\right)}d\epsilon\quad\text{for sufficiently small }\delta\\
    &\leq C\sqrt{N}\int_{\infty}^{\frac{1}{M\delta}} \frac{\sqrt{\log u}}{M^2u^2} du\\
    &\leq CM\delta \sqrt{N\log\left(\frac{1}{M\delta}\right)}
\end{aligned}
\end{equation}
\end{proof}
We can see for the linear space $\mathcal{F}_N$, the localized uniform entropy is basically $O(\sqrt{N}\delta)$ (if we omit the $\sqrt{\log(1/\delta)}$ term). When we construct the online projection estimator, the dimension of hypothesis space $N$ increases with sample size (we can also call $\mathcal{F}_N$ a sieve). As we will see later, the local diameter $\delta = \delta_n$ we consider decreases to zero at rate $\Theta(n^{-\frac{\alpha}{2\alpha+d}})$.

We use $\epsilon_i = Y_i - g_{\rho}(X_i),i=1,2,..,n$ to denote the i.i.d zero-mean noise variables and use $e_i, i=1,2,..,n$ to denote $n$ i.i.d. Rademacher variable, that is $\mathbb{P}(e_1 = 1) = \mathbb{P}(e_1 = -1) = \frac{1}{2}$.

In the following Proposition we require the noise to have a finite $\|\epsilon_i\|_{m,1}$-moment, which is defined as
\begin{equation}
\|\epsilon\|_{m, 1} := \int_{0}^{\infty} \mathbb{P}(|\epsilon|>t)^{1 / m} \mathrm{d} t
\end{equation}
Let $\Delta>0$, it is known that if $\epsilon_1$ has a finite $m+\Delta$-th moment, then it has a finite $\|\cdot\|_{m,1}$-moment \cite[Chapter~10]{ledoux2013probability}. So requiring having a finite $\|\cdot\|_{m,1}$, as assumed in (A1), is only slightly stronger than requiring a finite $m$-th moment.

Now we state and prove a proposition that connects the bounds on the multiplier/Rademacher process to the convergence rate of our M-estimator. This proposition is essentially the same as Theorem~3.4.1 in \cite{van1996weak} and is a slight generalization of Proposition~2 in \cite{han2019convergence}. In Proposition~\ref{prop:mainconnection}, for better presentation we drop the subscript of $\mathcal{F}_N$ and simply denote it as $\mathcal{F}$. But we should keep in mind that $\mathcal{F}$ is a function space that depends on $n$.
\begin{proposition}
\label{prop:mainconnection}
Denote $\mathcal{F}-f_{\rho} := \{f-f_{\rho}\ |\ f\in\mathcal{F}\}$ and $\mathcal{F} - f_{N} := \{f-f_N\ |\ f\in\mathcal{F}\}$. Assume $(\mathcal{F}-f_{\rho})\bigcup(\mathcal{F}-f_N)$ has an envelope function $F(x)\leq 1$. Let $X_i\stackrel{i.i.d.}{\sim}\rho_X$ and assume $\epsilon_i$ are i.i.d. with finite $\left\|\epsilon_1\right\|_{m,1}$-norm for some $m > 1$.  Assume that for any $\delta\geq 0$, for each $f^*\in\{f_{\rho}, f_N\}$,
\begin{equation}
\label{eq:localmaxnoise}
\mathbb{E} \sup _{f \in \mathcal{F}:\left\|f-f^*\right\|_2 \leq  \delta}\left|\frac{1}{\sqrt{n}} \sum_{i=1}^{n} \epsilon_{i}\left(f-f^*\right)\left(X_{i}\right)\right| = O\left(\phi_{n}(\delta)\right)
\end{equation}
and
\begin{equation}
\label{eq:localmaxrad}
\mathbb{E} \sup _{f \in \mathcal{F}:\left\|f-f^*\right\|_2 \leq  \delta}\left|\frac{1}{\sqrt{n}} \sum_{i=1}^{n} e_{i}\left(f-f^*\right)\left(X_{i}\right)\right| = O\left(\phi_{n}(\delta)\right)
\end{equation}
for some $\phi_{n}$ such that $\delta \mapsto \phi_{n}(\delta) / \delta$ is nonincreasing. Further assume that $\|f_N - f_{\rho}\|_2 \leq C\delta_n$.

Then 
\begin{equation}
    \left\|\hat f_{n,N}-f_N\right\|_2=O_P\left(\delta_{n}\right)
\end{equation}
for any $\delta_{n} \geq n^{-\frac{1}{2}+\frac{1}{2m}}$ such that $\phi_{n}\left(\delta_{n}\right) \leq \sqrt{n} \delta_{n}^{2}$. If $\epsilon_1$ has a finite $m$-th moment for some $m\geq 2$, then:
\begin{equation}
    \mathbb{E}\left[\left\|\hat f_{n,N}-f_N\right\|_2\right]=O\left(\delta_{n}\right)
\end{equation}
\end{proposition}

\begin{proof}
The proof is a slight generalization of Proposition~2 in \cite{han2019convergence}. The distance we are going to bound is not between $\hat f_{n,N}$ and $f_{\rho}$ but between $\hat f_{n,N}$ and $f_N$ (the population risk minimizer over $\mathcal{F}$). 
We first define a random process and its mean functional:
\begin{equation}
    \begin{aligned}
    &\mathbb{M}_{n} f := \frac{2}{n} \sum_{i=1}^{n}\left(f-f_{\rho}\right)\left(X_{i}\right) \epsilon_{i}-
    \frac{1}{n} \sum_{i=1}^{n}\left(f-f_{\rho}\right)^{2}\left(X_{i}\right)\\
    & M f := \mathbb{E}\left[\mathbb{M}_{n}(f)\right]=-P\left(f-f_{\rho}\right)^{2}
    \end{aligned}
\end{equation}
We have the following property of $M(\cdot)$. For any $f\in\{f\in\mathcal{F}\ |\ \|f-f_N\|_2 \geq 4\|f_N-f_{\rho}\|_2\}$, $Mf - Mf_N \leq -\frac{1}{4}\|f-f_N\|_2^2$. For the proof of this elementary inequality, see p.337 Exercise 5 in \cite{van1996weak}, taking their $x = f, y=f_N, z = f_{\rho}$. 

Our proof is a standard peeling argument. Let 
\begin{equation}
    \mathcal{F}_j:=\left\{f \in \mathcal{F}: 2^{j-1} t \delta_{n} \leq\left\|f-f_N\right\|_2< 2^jt\delta_n\right\}
\end{equation}
We choose a fixed $t$ large enough such that $t\delta_n \geq 4\|f_N-f_{\rho}\|_2$, we use the ERM property of $\hat f_{n,N}$:
\begin{equation}
    \begin{aligned}
        &\mathbb{P}\left(\left\|\hat f_{n,N}-f_N\right\|_2 \geq t \delta_{n}\right) \leq \sum_{j \geq 1} \mathbb{P}\left(\sup _{f \in \mathcal{F}_j}\left(\mathbb{M}_{n}(f)-\mathbb{M}_{n}\left(f_N\right)\right) \geq 0\right)\\
        &\leq \sum_{j\geq 1}\mathbb{P}\left(\sup _{f \in \mathcal{F}_j}\left(\mathbb{M}_{n}(f)-\mathbb{M}_{n}\left(f_N\right)-M(f)+M(f_N)\right)\geq 2^{2j-2}t^2\delta_n^2\right)\\
    \end{aligned}
\end{equation}
We write $\left(\mathbb{M}_{n}(f)-\mathbb{M}_{n}\left(f_N\right)-M(f)+M(f_N)\right)$ explicitly:
\begin{equation}
\begin{aligned}
    &\mathbb{M}_{n}(f)-\mathbb{M}_{n}\left(f_N\right)-M(f)+M(f_N)\\
    & = \frac{2}{n}\sum_{i=1}^n(f - f_N)(X_i)\epsilon_i +(P-\mathbb{P}_n)(f-f_{\rho})^2 + (\mathbb{P}_n - P)(f_N-f_{\rho})^2
\end{aligned}
\end{equation}
Then we can continue the peeling argument:
\begin{equation}
    \begin{aligned}
    &\mathbb{P}\left(\left\|\hat f_{n,N}-f_N\right\|_2 \geq t \delta_{n}\right)\\
    &\leq \sum_{j\geq 1}\mathbb{P}\left(\sup _{f\in\mathcal{F}:\|f-f_N\|_2\leq 2^jt\delta_n}\left|\frac{1}{\sqrt{n}}\sum_{i=1}^n(f-f_N)(X_i)\epsilon_i\right|\geq 2^{2j-5}t^2\sqrt{n}\delta_n^2\right)+\\
    &\mathbb{P}\left(\sup _{f\in\mathcal{F}:\|f-f_N\|_2\leq 2^jt\delta_n}\left|\frac{1}{\sqrt{n}}\sum_{i=1}^n(f-f_{\rho})^2(X_i) - \mathbb{E}(f-f_{\rho})^2\right|\geq 2^{2j-4}t^2\sqrt{n}\delta_n^2\right)+\\
    &\mathbb{P}\left(\left|\frac{1}{\sqrt{n}}\sum_{i=1}^n(f_N-f_{\rho})^2(X_i) - \mathbb{E}(f_N-f_{\rho})^2\right|\geq 2^{2j-4}t^2\sqrt{n}\delta_n^2\right)\\
    &\leq \sum_{j\geq 1}\mathbb{P}\left(\sup _{f\in\mathcal{F}:\|f-f_N\|_2\leq 2^jt\delta_n}\left|\frac{1}{\sqrt{n}}\sum_{i=1}^n(f-f_N)(X_i)\epsilon_i\right|\geq 2^{2j-5}t^2\sqrt{n}\delta_n^2\right)+\\
    &2\mathbb{P}\left(\sup _{f\in\mathcal{F}:\|f-f_N\|_2\leq 2^jt\delta_n}\left|\frac{1}{\sqrt{n}}\sum_{i=1}^n(f-f_{\rho})^2(X_i) - \mathbb{E}(f-f_{\rho})^2\right|\geq 2^{2j-4}t^2\sqrt{n}\delta_n^2\right)
    \end{aligned}
\end{equation}
The first term is the multiplier process that contains the noise variable $\epsilon_i$'s, for which we have bound (given by our assumptions). The second term can be related to the Rademacher process by standard symmetrization and contraction principles \citep{van1996weak}. There is still a miss-match between the supremum and the random variable to be bounded, to fix this we need to use the condition $\|f_N-f_{\rho}\|_2 \leq C\delta_n$:
\begin{equation}
\begin{aligned}
\|f-f_{\rho}\|_2&\leq\|f-f_N\| + \|f_N - f_{\rho}\|_2\\
&\leq \|f-f_N\|+C\delta_n\\
\Rightarrow \{f\in\mathcal{F}:\|f-f_N\|\leq 2^jt\delta_n\}
&\subset\{f\in\mathcal{F}:\|f-f_{\rho}\|_2\leq (2^{j}t+C)\delta_n\}
\end{aligned}
\end{equation}
Therefore the second term is bounded by
\begin{equation}
    2\mathbb{P}\left(\sup _{f\in\mathcal{F}:\|f-f_{\rho}\|_2\leq (2^jt+C)\delta_n}\left|\frac{1}{\sqrt{n}}\sum_{i=1}^n(f-f_{\rho})^2(X_i) - \mathbb{E}(f-f_{\rho})^2\right|\geq 2^{2j-4}t^2\sqrt{n}\delta_n^2\right)
\end{equation}
And the rest of the proof is the same as Proposition~2 in \cite{han2019convergence}.
\end{proof}
When $\epsilon_i$ is sub-Gaussian noise (note that sub-Gaussian/sub-exponential random variables have finite moments of all orders), the bound on the empirical process terms \eqref{eq:localmaxnoise} and \eqref{eq:localmaxrad} usually only depend on the entropy of $\mathcal{F}_N$: Thus the convergence rate will only depend on the entropy as well. However if we only assume moment conditions, then $\phi_n(\delta)$ will depend on both the entropy \emph{and} the moment order \cite[Lemma~9]{han2019convergence}: Thus the convergence rate would depend on both  as well when $m$ is not large enough. 

Now we state the following Lemma to complete our bound of $\mathbb{E}\|\hat f_{n,N} - f_N\|_2$. Its proof is postponed to after we conclude the main result.
\begin{lemma}
\label{lemma:multiplierprocess}
Assume (A1) and $\hat f_{n,N}\in \mathcal{F}_N$ defined in \eqref{eq:defineFN}. We select $N = \Theta\left( n^{\frac{d}{2\alpha + d}}\right)$. (Recall that $\alpha$ is the smoothness parameter, $d$ is the dimension of $X_i$ and $m$ is the moment index of $\epsilon_i$)

Then with $\delta_{n} = \Theta\left( n^{-\frac{\alpha}{2\alpha+d}} \vee n^{-\frac{1}{2}+\frac{1}{2m}}\right)$, for each $f^*\in\{f_N,f_{\rho}\}$ we have
\begin{equation}
\label{eq:multiplierprocess}
\begin{aligned}
\mathbb{E} \sup _{f\in\mathcal{F}_N:\|f-f^*\|_2 \leq \delta_{n}}&\left|\sum_{i=1}^{n} \epsilon_{i} \left(f-f^*\right)\left(X_{i}\right)\right|
 \vee \mathbb{E} \sup _{f\in\mathcal{F}_N:\|f-f^*\|_2 \leq \delta_{n}}\left|\sum_{i=1}^{n} e_{i} \left(f - f^*\right)\left(X_{i}\right)\right|\\
& \leq C_{\alpha}\left\{
\begin{array}{ll}
n^{\frac{d}{2\alpha+d}}\sqrt{\log n}\left(1 \vee\left\|\epsilon_{1}\right\|_{2\alpha+1, 1}\right), & m \geq 2\alpha/d + 1\\
n^{\frac{1}{m}}\sqrt{\log n}\left(1 \vee\left\|\epsilon_{1}\right\|_{m, 1}\right), & 1 \leq m< 2\alpha/d + 1
\end{array}
\right.
\end{aligned}
\end{equation}    
where $\|\epsilon_1\|_{2\alpha+1}$ is the $2\alpha+1$-th moment of $\epsilon_1$.
\end{lemma}
In light of Proposition~\ref{prop:mainconnection}, \eqref{eq:multiplierprocess}
can be written as
\begin{equation}
\begin{aligned}
&\mathbb{E} \sup _{f\in\mathcal{F}_N:\|f-f^*\|_2 \leq \delta_{n}}\left|\frac{1}{\sqrt{n}}\sum_{i=1}^{n} \epsilon_{i} \left(f - f^*\right)\left(X_{i}\right)\right|
\\&\quad\vee\ 
\mathbb{E} \sup _{f\in\mathcal{F}_N:\|f-f^*\|_2 \leq \delta_{n}}\left|\frac{1}{\sqrt{n}}\sum_{i=1}^{n} e_{i} \left(f - f^*\right)\left(X_{i}\right)\right|
\leq \phi_n(\delta_n)
\end{aligned}
\end{equation}
where 
\begin{equation}
    \phi_n(\delta) = 
    \left\{
\begin{array}{ll}
C_{\alpha}\sqrt{\log n/n}\delta^{-1/\alpha}\left(1 \vee\left\|\epsilon_{1}\right\|_{1+2\alpha, 1}\right), & m \geq 1+2\alpha \\
C_{\alpha}\sqrt{\log n/n}\delta^{-2/(m-1)}\left(1 \vee\left\|\epsilon_{1}\right\|_{m, 1}\right), & 1 \leq m<1+2\alpha
\end{array}
\right.
\end{equation}
\begin{lemma}
\label{lemma:lemmaMestimator}
Assume (A1) and $\hat f_{n,N}\in \mathcal{F}_N$. Choosing $N = \Theta( n^{\frac{d}{2\alpha + d}})$,
\begin{equation}
\label{mainresult}
    E[\|\hat f_{n,N} - f_N\|_2] = O\left(n^{-\frac{\alpha}{2\alpha+d}} \sqrt{\log n}\vee n^{-\frac{1}{2}+\frac{1}{2 m}} \sqrt{\log n}\right)
\end{equation}
\end{lemma}
\begin{proof}
We use the result of Lemma~\ref{lemma:multiplierprocess} as conditions of Proposition~\ref{prop:mainconnection}, and then identfy the smallest $\delta_n$ satisfying $\phi_n(\delta_n)\leq \sqrt{n}\delta_n^2$, which will give the stated convergence rate.
\end{proof}
\begin{proof}[Proof of Theorem~3]
We need only combine the bounds in Lemma~\ref{lemma:lemmabias} and Lemma~\ref{lemma:lemmaMestimator} using the triangle inequality.
\end{proof}
We now return to proving Lemma~\ref{lemma:multiplierprocess}. We first state two results, Propositions~\ref{prop:entropybound}, and~\ref{prop:boundbyhan}, from the literature which we will use to prove our Lemma. We begin with a standard result connecting Rademacher complexity and the entropy integral.
\begin{proposition}[Theorem~2.1, \cite{van2011local}]
\label{prop:entropybound}
Suppose that $\mathcal{G}$ has a finite envelope $G(x)\leq 1$ and $X_{1}, \ldots, X_{n}$ 's are i.i.d. random variables with law $P$. 

Then with $\mathcal{G}(\delta):=\left\{g \in \mathcal{G}: P g^{2}<\delta^{2}\right\}$,
\begin{equation}
\mathbb{E}\sup_{g\in\mathcal{G(\delta)}}\left|\frac{1}{\sqrt{n}}\sum_{i=1}^{n} e_{i} g\left(X_{i}\right)\right | = O\left( J\left(\delta, \mathcal{G}, L_{2}\right)\left(1+\frac{J\left(\delta, \mathcal{G}, L_{2}\right)}{\sqrt{n} \delta^{2}\|G\|_{P, 2}}\right)\|G\|_{P, 2}\right)
\end{equation}
\end{proposition}
We next give a recent inequality established in \cite{han2019convergence}. This allows us to relax common subgaussian assumptions to only moment conditions on the $\epsilon_i$'s.

\begin{proposition}[Theorem~1,\cite{han2019convergence}]
\label{prop:boundbyhan}
Suppose $X_i$'s, $\epsilon_i$'s are all i.i.d. random variables and $X_i$'s are independent of $\epsilon_i$'s. Let $\left\{\mathcal{G}_{k}\right\}_{k=1}^{n}$ be a sequence of function classes such that $\mathcal{G}_{k} \supset \mathcal{G}_{n}$ for any $1 \leq k \leq n .$ Assume further that there
exists a nondecreasing concave function $\psi_{n}: \mathbb{R}_{\geq 0} \rightarrow \mathbb{R}_{\geq 0}$ with $\psi_{n}(0)=0$ such that
\begin{equation}
\mathbb{E}\sup_{f\in \mathcal{G}_{k}}\left |\sum_{i=1}^{k} e_{i} f\left(X_{i}\right)\right | \leq \psi_{n}(k)
\end{equation}
holds for all $1\leq k\leq n$. Then
\begin{equation}
\mathbb{E}\sup_{f\in\mathcal{G}_{n}}\left |\sum_{i=1}^{n} \epsilon_{i} f\left(X_{i}\right)\right | \leq 4 \int_{0}^{\infty} \psi_{n}\left(\sum_{i=1}^{n} \mathbb{P}\left(\left|\epsilon_{i}\right|>t\right)\right) \mathrm{d} t
\end{equation}
\end{proposition}

With these two results in hand, we are now ready to prove Lemma~\ref{lemma:multiplierprocess}.
\begin{proof}[Proof of Lemma~\ref{lemma:multiplierprocess}]
We need to show the result for both $f^{*} = f_N$ and $f^{*} = f_{\rho}$. We will explicitly show the result for $f^* = f_N$: The proof in the case $f^* = f_{\rho}$ is exactly the same. 

Denote 
\begin{equation}
    \mathcal{F}_N(\delta_k):=\{f\in\mathcal{F}_N\ |\ \|f-f_N\|_2^2\leq \delta_k^2\}
\end{equation}
We first combine Proposition~\ref{prop:entropybound} with the entropy bound we established in Lemma~\ref{lemma:entropybound} to derive
\begin{equation}
\label{eq:general}
   \mathbb{E}\sup_{f\in \mathcal{F}_N(\delta_k)}\left |\sum_{i=1}^{k} e_{i} f\left(X_{i}\right)\right |  \leq C\delta_k k^{\frac{d}{2(2\alpha+d)}+\frac{1}{2}} \sqrt{\log k}
\end{equation}
where $\delta_k = k^{-\frac{\alpha}{2\alpha+d}}\vee k^{-\frac{1}{2}+\frac{1}{2m}}$. 

When $m\geq 2\alpha/d+1$ (recall $m$ is the moment index for $\epsilon_i$'s), $k^{-\frac{\alpha}{2\alpha + d}} > k^{-\frac{1}{2}+\frac{1}{2m}}$, so the above bound becomes
\begin{equation}
\label{eq:case1}
   \mathbb{E}\sup_{f\in \mathcal{F}_N(\delta_k)}\left |\sum_{i=1}^{k} e_{i} f\left(X_{i}\right)\right |  \leq C k^{\frac{d}{2\alpha+d}}\sqrt{\log k}
\end{equation}
Using \eqref{eq:case1} we see that the conditions of Proposition~\ref{prop:boundbyhan} are satisfied, thus giving us
\begin{equation}
\begin{aligned}
\mathbb{E}\sup_{f\in \mathcal{F}_N(\delta_k)}\left |\sum_{i=1}^{n} \epsilon_{i} f\left(X_{i}\right)\right | 
&\leq C\int_{0}^{\infty} \left(\sum_{i=1}^{n} \mathbb{P}\left(\left|\epsilon_{i}\right|>t\right)\right)^{\frac{d}{2\alpha+d}}\sqrt{\log\left(\sum_{i=1}^{n} \mathbb{P}\left(\left|\epsilon_{i}\right|>t\right)\right)} \mathrm{d} t\\
&= C n^{\frac{d}{2\alpha+d}}\sqrt{\log n}(1\vee \|\epsilon_1\|_{2\alpha+1,1})
\end{aligned}
\end{equation}
Note that we used $\epsilon_i$'s are i.i.d. random variables.

When $1<m<2\alpha/d + 1$, \eqref{eq:general} becomes
\begin{equation}
   \mathbb{E}\sup_{f\in \mathcal{F}_N(\delta_k)}\left |\sum_{i=1}^{k} e_{i} f\left(X_{i}\right)\right |  \leq C k^{\frac{1}{m}}\sqrt{\log k}.
\end{equation}
Plugging this in to Proposition~\ref{prop:boundbyhan} we get
\begin{equation}
\mathbb{E}\sup_{f\in \mathcal{F}_N(\delta_k)}\left |\sum_{i=1}^{n} \epsilon_{i} f\left(X_{i}\right)\right | 
\leq C n^{\frac{1}{m}}\sqrt{\log n}(1\vee \|\epsilon_1\|_{m,1})
\end{equation}
This compeltes the proof.
\end{proof}
\section{Online Projection Estimator and Functional Stochastic Gradient Descent}
\label{ComparisonMnSGD}
The computational expense of \textbf{Algorithm 2} is a dramatic improvement compared with SGD based algorithms, whose expense is $O(n)$ per updating. We also note that the computational expense of \textbf{Algorithm 2} depends on our assumption of the spectrum of operator $T_K$. The larger $\alpha$ is, the stronger our statistical assumption is, the faster our algorithm is. However, the expense of SGD-based algorithm is not sensitive to the statistical assumptions.\\\\
In this section we use the same notation as in Section 3 in the main text. We define $\hat{\pmb{\theta}}_{N,n}$ as the minimizer of the empirical loss
\begin{equation}\label{NPLS}
\min_{\pmb{\theta}\in\mathbb{R}^{N}} \sum_{i=1}^n (Y_i - \pmb{\theta}^{\top} \pmb{\psi}^N(X_i))^2
\end{equation}
Here we use double subscript to emphasize that $\hat{\pmb{\theta}}_{N,n}$ is calculated with $N$ basis function and $n$ data. Similarly, we can define $\hat{\pmb{\theta}}_{N,n-1}$ as the minimizer when there is one less sample $(X_n,Y_n)$ (but keep the other samples the same). There is actually a recursive relationship between $\hat{\pmb{\theta}}_{N,n}$ and $\hat{\pmb{\theta}}_{N,n-1}$:
\begin{equation}
\label{recursiveforOPE}
\hat{\pmb{\theta}}_{N,n}=\hat{\pmb{\theta}}_{N,n-1}+\Phi_{n} \pmb{\psi}_{n}\left[Y_{n}-\hat f_{n-1,N}(X_{n})\right]
\end{equation}
See \cite{ljung1983theory} p.18-20 for the derivation. This formula tells us how $\hat{\pmb{\theta}}_{N,n}$ changes when one additional data-pair is observed. If we see $\hat{\pmb{\theta}}_{N,n}$ as an update of $\hat{\pmb{\theta}}_{N,n-1}$ with $(X_n,Y_n)$, the step size will scale in proportion to the prediction error $|Y_n - \hat f_{n-1,N}(X_n)|$, and the direction is $\Phi_n\pmb{\psi}_{n}$ (which, in general, is not equal to $\pmb{\psi}_{n}$)\\\\
Similarly, we can derive a recursive relationship for how $\hat{\pmb\theta}_{N,n}$ changes when one more basis function $\psi_{N+1}$ is added in. Specifically,
\begin{equation}
\label{recursive}
 \hat{\pmb\theta}_{N+1,n}= \left[\begin{array}{c}{\hat{\pmb{\theta}}_{N, n}} \\ {0}\end{array}\right] + \frac{\left(\pmb{\psi}^{N+1}\right)^{\top} \pmb{\Delta}_n}{\left\|\left(I-P_{n}\right) \pmb{\psi}^{N+1}\right\|^{2}}
 \left[\begin{array}{c}{-P_n \pmb{\psi}^{N+1}} \\ {1}\end{array}\right] 
\end{equation}
Where $\pmb{\Delta}_n$ is the residual vector, whose i-th component is defined by:
\begin{equation}
    \pmb{\Delta}_n^{(i)} = Y_i - \hat f_{n,N}(X_i)
\end{equation}
and $P_n=(\Psi_n^{\top}\Psi_n)^{-1}\Psi_n^{\top}$ is the projection matrix of the column space of design matrix $\Psi_{n}$ with $N$ features. We give the derivation in the later part of this section.\\\\
The influence of a new feature on the regression coefficients is quantitatively associated with how much the residual can be explained by the new feature (represented by the term $\left(\pmb{\psi}^{N+1}\right)^{\top} \pmb{\Delta}_n$) and how orthogonal the new feature is to the old features (represented by $P_n \pmb{\psi}^{N+1}$).

However, if we use parametric stochastic gradient descent to solve the problem \eqref{NPLS}, then the updating rule should be:
\begin{equation}
\label{recursiveSGD}
\hat{\pmb{\theta}}_{N,n}=\hat{\pmb{\theta}}_{N,n-1}+\epsilon_n \pmb{\psi}_{n}\left[Y_{n}-\hat f_{n-1,N}(X_{n})\right]
\end{equation}
where we usually choose $\epsilon_n \asymp\frac{1}{n}$. 

Comparing \eqref{recursiveSGD} with \eqref{recursiveforOPE}, we see that it replaces the structured matrix $\Phi_n$ with a diagonal matrix $\epsilon_n I$. By doing so it omits the information of the correlation between features, this can help to illustrate why the SGD-based estimator \eqref{recursiveSGD} usually has a larger generalization error than the empirical risk minimizer \eqref{recursiveforOPE}.
\subsection{Proof of recursive formula \eqref{recursive}}
\begin{proof}
In this proof, we use a double subscript to indicate the dimension of the matrices. By definition of OLS estimator:
\begin{align*}
\hat{\pmb{\theta}}_{N+1,n}&=\Phi_{(N+1) \times(N+1)} \cdot \Psi_{n \times(N+1)}^{\top} \cdot \pmb{Y}_{n}\\
    & =\Phi_{(N+1) \times(N+1)} \cdot\left(\sum_{i=1}^{n} Y_{i}\left[\psi_{1}\left(X_{i}\right), \ldots, \psi_{N+1}\left(X_{i}\right)\right]^{\top}\right)\\
    &=\Phi_{(N+1) \times(N+1)} \cdot\left[\begin{array}{cc}{\sum_{i=1}^{n} \pmb{\psi}_{N}\left(X_{i}\right) Y_{i}} & {} \\ {\sum_{i=1}^{n} \psi_{N+1}\left(X_{i}\right) Y_{i}}\end{array}\right]\\
    &\stackrel{(1)}{=}\Phi_{(N+1) \times(N+1)} \cdot\left[
    \begin{array}{c}
    {\Phi_{N \times N}^{-1} \cdot \hat{\pmb{\theta}}_{N,n}} \\ 
    {\sum_{i=1}^{n} \psi_{N+1}\left(X_{i}\right) Y_{i}}
    \end{array}
    \right]\\
    &\stackrel{(2)}{=}
    \left(\left[
    \begin{array}{cc}
    {\Phi_{N \times N}} & {0} \\ 
    {0} & {0}
    \end{array}
    \right]+A\right)\cdot\left[
    \begin{array}{c}
    {\Phi_{N \times N}^{-1} \cdot \hat{\pmb{\theta}}_{N,n}} \\ 
    {\sum_{i=1}^{n} \psi_{N+1}\left(X_{i}\right) Y_{i}}
    \end{array}
    \right]\\
\end{align*}
where
$$
A = \left[\begin{array}{cc}{\frac{1}{k} \Phi_{n-1}\ \pmb{b} \pmb{b}^{T}\ \Phi_{n-1}} & {-\frac{1}{k} \Phi_{n-1}\ \pmb{b}} \\ {-\frac{1}{k} \pmb{b}^{T} \ \Phi_{n-1}} & {\frac{1}{k}}\end{array}\right]
$$
$$
\pmb{b}= \Psi_{n-1}^T \pmb{\psi}_{N+1}
$$
$$
k = \pmb{\psi}_{N+1}^T\pmb{\psi}_{N+1} - \pmb{b}^T \Phi_{n-1}\ \pmb{b}
$$
In (1) we use the definition of $\hat{\pmb{\theta}}_{N,n}$ and in (2) use the block matrix inversion formula.
\begin{equation}
    \hat{\pmb{\theta}}_{N+1,n} =\left[\begin{array}{c}{\hat{\pmb{\theta}}_{N,n}} \\ {0}\end{array}\right] + \frac{1}{k} \cdot\left[\begin{array}{c}{\Phi_{N \times N} \pmb{b}\left(\pmb{b}^{T} \hat{\pmb{\theta}}_{N,n}-\sum_{i=1}^{n}\psi_{N+1}\left(X_{i}\right) Y_{i}\right)} \\ {\left(\sum_{i=1}^{n} \psi_{N+1}\left(X_{i}\right) Y_i-\pmb{b}^{\top} \hat{\pmb{\theta}}_{N,n}\right)}\end{array}\right]
\end{equation}
Note that 
\begin{equation}
    \pmb{b}^{\top} \hat{\pmb{\theta}}_{N,n} =\sum_{i=1}^{n} \psi_{N+1}\left(X_{i}\right) \sum_{j=1}^N \psi_j(X_i)\hat{\pmb{\theta}}_{N,n}^{(j)}= \sum_{i=1}^{n} \psi_{N+1}\left(X_{i}\right) \hat f_{n,N}(X_i)
\end{equation}
So
\begin{equation}
 \sum_{i=1}^{n} \psi_{N+1}\left(X_{i}\right) Y_{i}-\pmb{b}^{\top} \hat{\pmb{\theta}}_{N,n}=\sum_{i=1}^{n} \psi_{N +1}\left(X_{i}\right)\left(Y_{i}-f_{n,N}(X_i)\right)   
\end{equation}
Continuing, we see that
\begin{align*}
\hat{\pmb{\theta}}_{N+1,n}&=\left[\begin{array}{c}{\hat{\pmb{\theta}}_{N,n}} \\ {0}\end{array}\right] + \frac{\pmb{\psi}_{N+1}^{\top} \pmb{\Delta}_n}{k} \cdot\left[\begin{array}{c}{-\Phi_{N \times N} \pmb{b}} \\ {1}\end{array}\right]
\end{align*}
Now we expand $k$:
\begin{align*}
    k&=\pmb{\psi}_{N+1}^{\top} \pmb{\psi}_{N+1}-\pmb{\psi}_{N+1}^{\top} \Psi_{n\times N} \Phi_{N \times N} \Psi_{n \times N}^{\top} \pmb{\psi}_{N+1}\\
    &=\pmb{\psi}_{N+1}^{\top}\left(I-\Psi_{n \times N}\left(\Psi_{n \times N}^{\top} \Psi_{n \times N}\right)^{-1} \Psi_{n \times N}\right) \pmb{\psi}_{N+1}\\
    &=\left\|\left(I-P_{n}\right) \pmb{\psi}_{N+1}\right\|^{2}
\end{align*}
And use the definition of $b$:
\begin{equation}
    \begin{aligned}
     \hat{\pmb{\theta}}_{N+1,n}&=\left[\begin{array}{c}{\hat{\pmb{\theta}}_{N,n}} \\ {0}\end{array}\right]\\
&+\frac{\pmb{\psi}_{N+1}^{\top} \pmb{\Delta}_n}{\left\|\left(I-P_{n}\right) \pmb{\psi}_{N+1}\right\|^{2}}\left[\begin{array}{c}{-\Phi_{N \times N} \left[\begin{array}{c}{\pmb{\psi}_{1}^{\top} \pmb{\psi}_{N+1}} \\ {\vdots} \\ {\pmb{\psi}_{N}^{\top} \pmb{\psi}_{N+1}}\end{array}\right]} \\ {1}\end{array}\right] \\
& = \left[\begin{array}{c}{\hat{\pmb{\theta}}_{N,n}} \\ {0}\end{array}\right] + \frac{\pmb{\psi}_{N+1}^{\top} \pmb{\Delta}_n }{\left\|\left(I-P_{n}\right) \pmb{\psi}_{N+1}\right\|^{2}} \left[\begin{array}{c}{-P_n \pmb{\psi}_{N+1}} \\ {1}\end{array}\right]
    \end{aligned}
\end{equation}
\end{proof}

\section{Regression in Additive Models}
In the main text we discussed estimation in multivariate RKHS and how it suffers from the curse of dimensionality. For $X_i\in\mathbb{R}^d$, it is also quite common to impose an extra additive structure on the model, in other words, we assume
\begin{equation}
    f_{\rho}(x_i) = \sum_{k=1}^d f_{\rho,k}\left(x^{(k)}_i\right)
\end{equation}
where the component functions $f_{\rho,i}$ belong to a RKHS $\mathcal{H}$ (in general they can belong to different spaces), and $x^{(k)}_i$ is the k-th entry of $x_i$. Such a model is a generalization of the multivariate linear model. It balances modeling flexibility with tractability of estimation. See eg. \citet{hastie2009elements} and \citet{yuan2016minimax} for further discussion.

The projection estimator for an additive model is obtained by solving the following least-squares problem in Euclidean space (which is essentially the same as solving the problem~\eqref{NPLS}).
\begin{equation}\label{eq:add}
    \min_{\pmb{\theta}\in\mathbb{R}^{N\times d}}\sum_{i=1}^n (Y_i - \sum_{k=1}^d\sum_{j=1}^N\theta_{jk} \psi_j(x_i^{(k)}))^2
\end{equation}
here $N$ still needs to be chosen of order $n^{\frac{1}{2\alpha+1}}$, when $\lambda_j = \Theta(j^{-2\alpha})$. The online projection estimator in an additive model is
\begin{equation}
    \hat f_{n,N} = \sum_{k=1}^d \sum_{j=1}^N \hat\theta_{jk}\psi_j
\end{equation}
For a fixed $d$, the minimax rate for estimating an additive model is identical (losing a constant $d$) to the minimax rate in the analogous one-dimension nonparametric regression problem working with the same hypothesis space $\mathcal{H}$ \citep{raskutti2009lower}.

The design matrix of \eqref{eq:add} now is of dimension $n\times (Nd)$. When a new data point is collected, our design matrix grows by one row. When we need to increase the model capacity however, we need to add one feature for each dimension (in total $d$ columns). Updating such estimators when $X_i\in \mathbb{R}^d$ has a computational expense of order $O(d^2n^{\frac{2}{2\alpha+1}})$, by a argument similar to that presented in Section~3.4. To clarify, in Section~3.4 we are assuming the eigenvalue $\lambda_j = \Theta(j^{-2\alpha/d})$ (for example, the RKHS is $d$-dimension, $\alpha$-th order Sobolev space); however in this section we are discussing $d$-dimension additive model, \emph{each component} lies in a 1-dimension RKHS whose $\lambda_j = \Theta(j^{-2\alpha})$. The additive model is more restrictive, therefore we have better statistical and computational guarantee when the model is well-specified.

\subsection{Additive Model Application}
We chose a 10-dimension additive function to illustrate the efficacy of our method for fitting additive models. In this example, the components of the $f_{\rho}$ in each dimension are Doppler-like functions. For $x\in\mathbb{R}^{10}$,
\begin{equation}
   \begin{aligned}
    f_{\rho}(x) =&\sum_{k=1}^{10}f_{\rho,k}(x^{(k)})\\ =&\sum_{k=1}^{10}\left\{ \sin\left(\frac{2\pi}{(x^{(k)}+0.1)^{k/20}}\right) 
    - \sin\left(\frac{2\pi}{0.1^{k/20}}\right)\right\}
\end{aligned} 
\end{equation}

Similar functions are used in~\citet{sadhanala2019additive}. The kernel (for each dimension) we consider is 
\begin{equation}
    K(s,t) = \sum_{m=1}^2 s^mt^m+B_4(\{s-t\})
\end{equation}
In Figure~\ref{ADD}, we compare the method in this paper with the additive smoothing spline estimator calculated with back fitting using R package 'gam' \citep{gam}. Both of the methods achieve rate-optimal convergence, but we note the smoothing spline method takes dramatically more time as an offline estimator.

\begin{figure}[ht]
\vskip 0.2in
\begin{center}
\centerline{\includegraphics[width=\columnwidth]{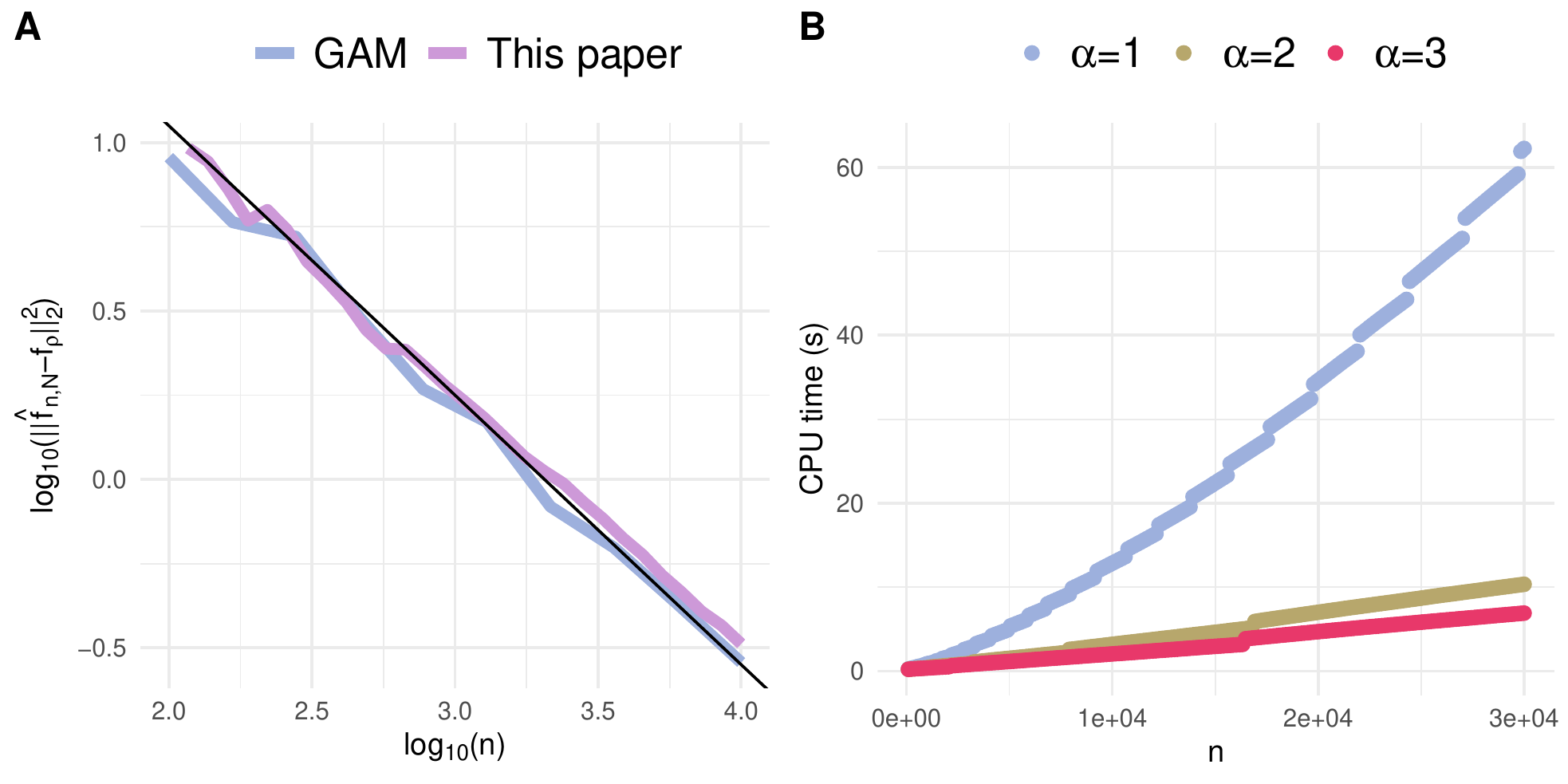}}
\caption{Additive model: generalization error and CPU time. \textbf{(A)} Both smoothing spline and online projection estimator achieve the optimal rate $O(n^{-4/5})$. The black line has slope $-4/5$. Each curve is based on 15 independent runs. \textbf{(B)} The CPU time decreases as $\alpha$ becomes larger (repetitions=10).}
\label{ADD}
\end{center}
\vskip -0.2in
\end{figure}

\section{Details of simulation studies}
In the main text we gave important details on of the settings of our simulation studies. To help our readers replicate our result, we now list all details for our simulations.
\subsection{Notation and general setting}
The $\|\hat f_{n,N}-f_{\rho}\|_2^2$ on the y-axis of Figure 2 is estimated with 1,000 $X$ generated from $\rho_X$. The estimator based on kernel ridge regression (KRR) is defined as the minimizer of penalized mean-square error
\begin{equation}
    \min _{f \in \mathcal{H}}\frac{1}{n} \sum_{i=1}^{n}\left(Y_{i}-f\left(X_{i}\right)\right)^{2}+\lambda_{n,KRR}\|f\|_{\mathcal{H}}^{2}
\end{equation}
for a closed form solution and theoretical optimal selection of $\lambda_{n,KRR}$, see 12.5.2 and Theorem 13.7 of \cite{wainwright2019high}.\\\\
In the main text, we slightly simplify the update rule for nonparametric SGD estimator without losing the essential principles. In all the simulation study of this paper, the SGD estimator we use is the version with Polyak averaging (p.1375-1376 of \cite{dieuleveut2016nonparametric})
\begin{equation}
\tilde f_{n}=\tilde f_{n-1}+\gamma_{n,SGD}\left[Y_{n}-\tilde f_{n-1}\left(X_{n}\right)\right] K_{X_{n}}
\end{equation}
\begin{equation}
\hat {f}_{n}=\frac{1}{n+1} \sum_{k=0}^{n} \tilde f_{k}
\end{equation}
The nonparametric SGD estimator we use is $\hat f_n$. To update such an estimator, the computational cost is also $O(n)$.

All the simulation study examples are coded in R version 3.5.1.
\subsection{One Dimension Example Settings}
We give the details of example 1 (resp. example 2) in Table~\ref{example 1} (resp. Table~\ref{example 2}).
\begin{table}[!htbp]
\caption{Settings of example 1. See \cite{wahba1990spline} and \cite{dieuleveut2016nonparametric}} \label{example 1}
\begin{center}
\begin{tabular}{ll}
\hline
$f_{\rho}$ & $B_4(x) = x^{4}-2 x^{3}+x^{2}-\frac{1}{30}$\\
$\epsilon$ & Unif([-0.02,0.02])\\
$p_X(x)$ &$\bm{1}_{[0,1]}(x)$\\
$K(s,t)$ & $\frac{-1}{24}B_4(\{s-t\})=\sum_{j=1}^{\infty} \frac{2}{(2 \pi j)^{4}}[\cos (2 \pi j s) \cos (2 \pi j t) $\\
&$\quad+\sin( 2 \pi j s) \sin( 2 \pi j t)]$\\
RKHS $\mathcal{H}$ & $W_2^{per}=\left\{f\in L^2([0,1])| \int_{0}^{1} f(u) d u=0, \right.$\\
& $\left.\quad\quad f(0)=f(1),f'(0)=f'(1),\int_{0}^{1} \left(f^{(2)}(u)\right)^2 d u<\infty\right\}$\\
$\lambda_j$ & $\frac{2}{(2 \pi j)^{4}} = O(j^{-4})$\\
$\psi_j(x)$ & $\sin(2\pi j x)$ and $\cos(2\pi j x)$\\
basis adding step & $n = \lfloor 0.2 N^{5}\rfloor$\\
Hyperparameter KRR $\lambda_{n,KRR}$ & $\lambda_{n,KRR}=10^{-3} n^{-4/5}$\\
Learning rate $\gamma_{n,SGD}$ & $\gamma_{n,SGD} = 128n^{-0.5}$
\end{tabular}
\hrule
\end{center}
\end{table}

\begin{table}[!htbp]
\caption{Settings of example 2. See \citet[Chap.~12]{wainwright2019high} for more discussion on the kernel space $W_1^0$.} 
\label{example 2}
\begin{center}
\hrule
\begin{tabular}{ll}
$f_{\rho}$ & $(6x-3)\sin(12x-6)+\cos^2(12x-6)$\\
$\epsilon$ & Normal(0,5)\\
$p_{\rho}(x)$&$(x+0.5)\bm{1}_{[0,1]}(x)$\\
$K(s,t)$ & $\min\{s,t\} = \sum_{j=1}^{\infty} \frac{8}{(2 j-1)^{2} \pi^{2}} \sin \left(\frac{(2 j-1) \pi s}{2}\right)\sin \left(\frac{(2 j-1) \pi t}{2}\right)$ \\
RKHS $\mathcal{H}$ & $W_1^0 = \left\{f\in L^2([0,1]) | f(0) = 0, \int_0^1(f'(u))^2 du < \infty\right\}$\\
$\lambda_j$ &$\frac{2}{(2j-1)^2\pi^2} = O(j^{-2})$ \\
$\psi_j(x)$ & $2\sin\left(\frac{(2j-1)\pi x}{2}\right)$\\
basis adding step & $n = \lfloor 0.5 N^{3}\rfloor$\\
Hyperparameter KRR $\lambda_{n,KRR}$ & $\lambda_{n,KRR}=0.1 n^{-2/3}$\\
Learning rate $\gamma_{n,SGD}$ & $\gamma_{n,SGD} = 5n^{-0.5}$\\
\end{tabular}
\hrule
\end{center}
\end{table}

\subsection{Additive Model Example}
We use the function gam() in R package $\textbf{gam}$ \cite{gam} to fit the additive model with smoothing spline. The degrees of freedom parameter used in the {\tt s()} function were selected to increase with $n$. The details for the additive model example (including parameter selection) are given in Table~\ref{additivesetting}.
\begin{table}[!htbp]
\caption{Settings of Additive model example.} 
\label{additivesetting}
\begin{center}
\hrule
\begin{tabular}{ll}
$f_{\rho}$ & $ \sum_{k=1}^{10}\left\{ \sin\left(\frac{2\pi}{(X^{(k)}+0.1)^{k/20}}\right) 
    - \sin\left(\frac{2\pi}{0.1^{k/20}}\right)\right\}$\\
$\epsilon$ & Normal(0,5)\\
$p_{\rho}(X_1,...,X_{10})$&$\Pi_{k=1}^{10}\bm{1}_{[0,1]}(X^{(k)})$\\
$K(s,t)$ (for each dimension) & $\sum_{m=1}^2 s^mt^m + B_4(\{s-t\})$ \\
RKHS $\mathcal{H}$ & $W_2 = \left\{f\in L^2([0,1])\ |\ \int_0^1(f''(u))^2du < \infty\right\}$\\
$\lambda_j$ &$\frac{2}{(2 \pi j)^{4}}=O\left(j^{-4}\right)$ \\
$\psi_j(x)$ & $x,x^2,\sin(2\pi jx),\cos(2\pi jx)$\\
basis adding step & $n = \lfloor 0.2 N^{5}\rfloor$\\
df for smoothing spline & $2\lfloor n^{1/5}\rfloor$\\
\end{tabular}
\hrule
\end{center}
\end{table}

\section{A Note for Application and Additional Examples}
The hypothesis spaces used so far in this paper have been well-studied in previous work, and are relatively easy to engage with: Their kernel functions have a closed form, and their eigenfunctions can also be explicitly written out with respect to some special measures $\bar\rho$.\\\\
However, they are usually equipped with some undesirable boundary conditions. For example, in example 2, it is more interesting to consider the space
\begin{equation}
    W_{1}=\left\{f \in L^{2}([0,1]) |  \int_{0}^{1}\left(f^{\prime}(u)\right)^{2} d u<\infty\right\}
\end{equation}
rather than the one we use in our simulation study
\begin{equation}
  W_{1}^{0}=\left\{f \in L^{2}([0,1]) | f(0)=0, \int_{0}^{1}\left(f^{\prime}(u)\right)^{2} d u<\infty\right\}  
\end{equation}
Although it is known that $W_1$ is also an RKHS \cite{wainwright2019high} with kernel $\tilde K(s,t) = 1+\min\{s,t\}$, it takes extra analytical work to get the form of eigenfunctions for $\tilde K$.\\\\
For practical purposes, it is enough to consider functions of the following form as estimator:
\begin{equation}
    \hat f_{n,N}(x) = \theta_0\cdot 1+\sum_{j=1}^N \theta_j\psi_j(x)
\end{equation}
where $\psi_j = \sqrt{2} \sin \left(\frac{(2 j-1) \pi x}{2}\right)$ as stated in Table~1. Because the difference between $W_1^0$ and $W_1$ is merely a constant function in the sense that
\begin{equation}
    W_1 = \{1\}\oplus W_1^0
\end{equation}
When a new sample comes in, we update $\hat f_{n,N}$ (and potentially add a new basis function) in an online manner as in Algorithm 2. Similarly, in example 1, the more interesting space is
\begin{equation}
    W_2 = \left\{f\in L^2([0,1])|\int_{0}^{1}\left(f^{(2)}(u)\right)^{2} d u<\infty \right\}
\end{equation}
Note that
\begin{equation}
    W_2 = \{1\}\oplus \{x\}\oplus\{x^2\}\oplus W_2^{per}
\end{equation}
So the projection estimator can be of the form
\begin{equation}
    \hat f_{n,N}(x) = \sum_{k=0}^2\tilde \theta_k x^k + \sum_{j=1}^N \theta_j \psi_j(x)
\end{equation}
where $\psi_j$'s are the trigonometric functions listed in Table~1.

\begin{figure}[!ht]
\begin{center}
\centerline{\includegraphics[width=\columnwidth]{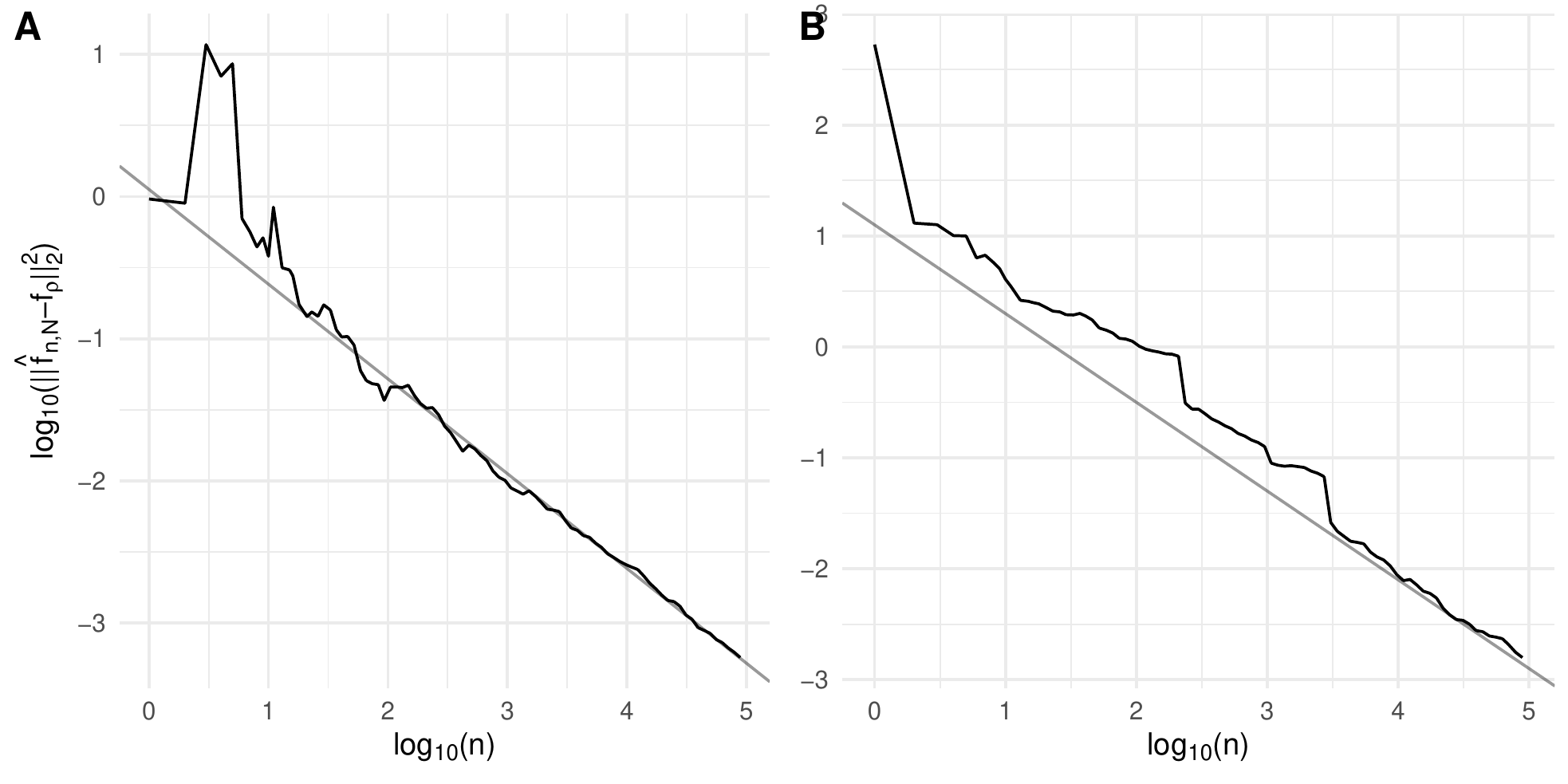}}
\caption{Generalization error for additional examples.\textbf{(A)} Example A.1, black line has slope $-2/3$ \textbf{(B)} Example A.2, the black line has slope $-4/5$. Both estimators achieve the minimax rates in $W_1$ and $W_2$. Each curve is based on 15 independent repetitions.}
\label{AdditionalSimu}
\end{center}
\vskip -0.2in
\end{figure}

The settings for our two additional examples are given in Table~\ref{additional}
\begin{table}[!htbp]
\caption{Settings of additional examples.} 
\label{additional}
\begin{center}
\begin{tabular}{lll}
\hline
&Example A.1 & Example A.2\\
\hline
$f_{\rho}$ & $1+(x-0.5)\bm{1}_{[0.5,1]}(x)$ &  $1+(6x-3)\sin(12x-6)+\cos^2(12x-6)$\\
&$\quad+2(x-0.2)\bm{1}_{[0.2,1]}(x)$& $\quad +10(x-0.5)^2\bm{1}_{[0.5,1]}(x)$\\
$\epsilon$ & Normal(0,1) & Unif(-5,5)\\
$p_{\rho}(x)$&$(x+0.5)\bm{1}_{[0,1]}(x)$&$\bm{1}_{[0,1]}(x)$\\
RKHS & $W_1$ & $W_2$\\
basis function & $1,\sin\left(\frac{(2j-1)\pi x}{2}\right),j=1,2,...$ & $1,x,x^2,\sin(2\pi jx),\cos(2\pi jx),j=1,2,...$\\
basis adding step & $n = \lfloor 0.5 N^{3}\rfloor$ & $n = \lfloor \frac{1}{30} N^{5}\rfloor$\\
\hline
\end{tabular}
\end{center}
\end{table}

\clearpage
\end{appendices}

\bibliographystyle{apalike}
\bibliography{main}

\end{document}